\documentclass[english,12pt, draftclsnofoot, peerreview, a4paper, oneside, onecolumn]{IEEEtran}
\usepackage[T1]{fontenc}
\usepackage[latin9]{inputenc}
\usepackage{xcolor}
\usepackage{pdfcolmk}
\usepackage{float}
\usepackage{mathrsfs}
\usepackage{url}
\usepackage{bm}
\usepackage{amsmath}
\usepackage{amsthm}
\usepackage{amssymb}
\usepackage{graphicx}
\PassOptionsToPackage{normalem}{ulem}
\usepackage{ulem}

\makeatletter

\floatstyle{ruled}
\newfloat{algorithm}{tbp}{loa}
\providecommand{\algorithmname}{Algorithm}
\floatname{algorithm}{\protect\algorithmname}
\providecolor{lyxadded}{rgb}{0,0,1}
\providecolor{lyxdeleted}{rgb}{1,0,0}
\DeclareRobustCommand{\lyxadded}[3]{{\color{lyxadded}{}#3}}
\DeclareRobustCommand{\lyxdeleted}[3]{{\color{lyxdeleted}\lyxsout{#3}}}
\DeclareRobustCommand{\lyxsout}[1]{\ifx\\#1\else\sout{#1}\fi}

\theoremstyle{plain}
\newtheorem{lem}{\protect\lemmaname}
\theoremstyle{plain}
\newtheorem{cor}{\protect\corollaryname}
\theoremstyle{plain}
\newtheorem{prop}{\protect\propositionname}
\theoremstyle{plain}
\newtheorem{thm}{\protect\theoremname}

\usepackage{amsmath}
\usepackage{amssymb}

\usepackage{graphicx,psfrag,cite,subfigure}

\author{

\IEEEauthorblockN{Wenjie~Liu and Junting~Chen}

\IEEEauthorblockA{School of Science and Engineering (SSE) \\
Future Network of Intelligence Institute (FNii) \\ 
The Chinese University of Hong Kong, Shenzhen, Guangdong 518172, China}

}


\usepackage[acronym]{glossaries}
\newcommand{\newac}{\newacronym}
\newcommand{\ac}{\gls}
\newcommand{\Ac}{\Gls}
\newcommand{\acpl}{\glspl}

\makeglossaries

\newac{speb}{SPEB}{square position error bound}
\newac[plural=EFIMs,firstplural=Fisher information matrices (EFIMs)]{efim}{EFIM}{Fisher information matrix}
\newac{ne}{NE}{Nash equilibrium}
\newac{mse}{MSE}{mean squared error}
\newac{toa}{TOA}{time-of-arrival}
\newac{snr}{SNR}{signal-to-noise ratio}
\newac{lan}{LAN}{local area network}
\newac{psd}{PSD}{positive semidefinite}
\newac{pd}{PD}{positive definite}
\newac{wrt}{w.r.t.}{with respect to}
\newac{lhs}{L.H.S.}{left hand side}
\newac{wp1}{w.p.1}{with probability 1}
\newac{kkt}{KKT}{Karush-Kuhn-Tucker}
\newac{wlog}{w.l.o.g.}{without loss of generality}
\newac{mle}{MLE}{maximum likelihood estimation}
\newac{gps}{GPS}{global positioning system}
\newac{rssi}{RSSI}{received signal strength indicator}
\newac{mimo}{MIMO}{multiple-input multiple-output}
\newac{csi}{CSI}{channel state information}
\newac{fdd}{FDD}{frequency division duplexing}
\newac{ms}{MS}{mobile station}
\newac{bs}{BS}{base station}
\newac{d2d}{D2D}{device-to-device}
\newac{slnr}{SLNR}{signal-to-interference-leakage-and-noise-ratio}
\newac{ula}{ULA}{uniform linear antenna array}
\newac{pas}{PAS}{power angular spectrum}
\newac{mmse}{MMSE}{minimum mean square error}
\newac{zf}{ZF}{zero-forcing}
\newac{rzf}{RZF}{regularized zero-forcing}
\newac{as}{AS}{angular spread}
\newac{aod}{AOD}{angle of departure}
\newac{iid}{i.i.d.}{independent and identically distributed} 
\newac{sinr}{SINR}{signal-to-interference-and-noise ratio}
\newac{tdd}{TDD}{time-division duplex}
\newac{rvq}{RVQ}{random vector quantization}
\newac{rhs}{R.H.S.}{right hand side}
\newac{mrc}{MRC}{maximum ratio combining}
\newac{cdf}{CDF}{cumulative distribution function}
\newac{a.s.}{a.s.}{almost surely}
\newac{los}{LOS}{line-of-sight}
\newac{jsdm}{JSDM}{joint spatial division and multiplexing}
\newac{map}{MAP}{maximum a posteriori}
\newac{klt}{KLT}{Karhunen-Lo\`eve Transform}
\newac{lbe}{LBE}{link bargaining equilibrium}
\newac{se}{SE}{Stackelberg equilibrium}
\newac{uav}{UAV}{unmanned aerial vehicle}
\newac{nlos}{NLOS}{non-line-of-sight}
\newac{pdf}{PDF}{probability density function}
\newac{em}{EM}{expectation-maximization}
\newac{knn}{KNN}{$k$-nearest neighbor}
\newac{svd}{SVD}{singular value decomposition}
\newac{nmf}{NMF}{non-negative matrix factorization}
\newac{umf}{UMF}{unimodality-constrained matrix factorization}
\newac{rmse}{RMSE}{rooted mean squared error}
\newac{olos}{OLOS}{obstructed line-of-sight}
\newac{mmw}{mmW}{millimeter wave}
\newac{ber}{BER}{bit error rate}
\newac{rss}{RSS}{received signal strength}
\newac{lp}{LP}{linear program}
\newac{ufw}{U-FW}{unimodal Frank-Wolfe}
\newac{utf}{UTF}{unimodality-constrained tensor factorization}
\newac{fw}{FW}{Frank-Wolfe}
\newac{iot}{IoT}{Internet-of-Things}
\newac{mae}{MAE}{mean absolute error}
\newac{crb}{CRB}{Cram\'er-Rao bound}
\newac{aoa}{AoA}{angle of arrival}
\newac{wcl}{WCL}{weighted centroid localization}

\setkeys{Gin}{width=1.0\columnwidth}

\renewcommand{\lyxdeleted}[3]{{\color{lyxdeleted}{}}}
\renewcommand{\lyxadded}[2]{{\color{black}{}}}

\makeatother

\usepackage{babel}
\providecommand{\corollaryname}{Corollary}
\providecommand{\lemmaname}{Lemma}
\providecommand{\propositionname}{Proposition}
\providecommand{\theoremname}{Theorem}

\begin{document}
\title{UAV-aided Radio Map Construction Exploiting Environment Semantics}

\maketitle
%
%


\newcommand*{\SINGLECOLUMN}{}

\ifdefined\SINGLECOLUMN
	\setkeys{Gin}{width=0.5\columnwidth}
	\newcommand{\figfontsize}{\footnotesize} 
\else
	\setkeys{Gin}{width=1.0\columnwidth}
	\newcommand{\figfontsize}{\normalsize} 
\fi

\begin{abstract}
This paper constructs a full dimensional (6D) radio map to predict
the channel gain between any transmitter location and any receiver
location based on \ac{rss} measurements between low-altitude aerial
nodes and ground nodes. The main challenge is how to describe the
signal strength attenuation due to the blockage from the environment.
Conventional interpolation-type approaches fail to exploit the close
relation between the radio map and the geometry of the environment.
This paper proposes to construct radio maps by first estimating and
constructing a multi-class 3D virtual obstacle map that describes
the geometry of the environment with radio semantics. Mathematically,
a least-squares problem is formulated to jointly estimate the virtual
obstacle map and the propagation parameters. This problem is found
to have a partial quasiconvexity that leads to the development of
an efficient parameter estimation and radio map construction algorithm.
Numerical experiments confirm that the proposed method substantially
reduces the amount of measurement required for achieving the same
radio map accuracy. It is also demonstrated that in a \ac{uav}-aided
relay communication scenario, a radio map assisted approach for \ac{uav}
placement can achieve more than 50\% capacity gain.
\end{abstract}


\glsresetall

\section{Introduction}

It becomes increasingly important for wireless communication networks
to learn about the communication environment. For example, millimeter-wave,
terahertz, and integrated aerial and terrestrial communications favor
a \ac{los} propagation condition. Thus, knowing the radio environment
and even the geometry of the surrounding may help user selection,
beamforming, and position optimization for communication nodes \cite{LimChoSimKim:J20,MoHuaXu:J19,ZhaZha:J20,EsrGanGes:J20}.
A recent trend is to leverage \emph{radio maps} for optimizing communication
networks, where a radio map is a data model that captures the location-dependent
wireless channel quality between a transmitter and a receiver \cite{ZenXu:M21,XiaWanXuXu:J22,ZenXuJinZha:J21,HuCaiLiuYu:J20}.

Recent studies have exploited radio maps for \ac{uav} position optimization
and trajectory planning for data harvesting, blockage-aware wireless
power transfer, and network localization \cite{MoHuaXu:J19,ZhaZha:J20,EsrGanGes:J20,ZenXu:M21,XiaWanXuXu:J22,ZenXuJinZha:J21,HuCaiLiuYu:J20}.
For example, in a typical scenario of \ac{uav}-to-ground communication
in dense urban environment, there are buildings and trees that probably
appear at arbitrary locations and block the air-to-ground signal.
Most earlier works addressed this issue using a \emph{probabilistic}
model, which describes the probability of the \ac{los} condition
as a function of the elevation angle at the ground node \cite{AlhKanJam:C14,YouZha:J20}.
By contrast, radio map based models can adapt to specific geography
environments and determine the actual location-dependent \ac{los}
condition and channel gain. Some recent study on the \ac{uav} relay
communication further revealed that, when a radio map is partially
available, the throughput performance of a position-optimized \ac{uav}
relay network can be substantially enhanced as compared to the methods
based on probabilistic models \cite{MoHuaXu:J19,ZenXuJinZha:J21,HuCaiLiuYu:J20}.

However, little is known on how to efficiently construct a radio
map. The following challenges need to be addressed. First, it usually
requires a huge amount of measurement data for radio map construction
due to the ample degrees of freedom for a propagation channel. A \emph{full}
dimensional radio map for a narrowband single antenna system may still
need 6 dimensions to describe the channel quality between any transmitter
location and any receiver location in 3D. Second, it is also costly
to store, transfer, and share radio maps among communication nodes.
Third, it is essential, yet challenging, to embed the environment
information to radio maps as wireless channels depend on the geography
environment via a complicated mechanism involving signal reflection,
diffraction, and scattering. As to be discussed below, pure data-driven
environment-blind approaches may result in low efficiency of utilizing
the measurement data; classical channel models, such as the simplified
probabilistic \ac{los} model, may suffer from poor prediction performance;
and ray-tracing methods suffer from the overwhelming requirement for
computing capability and high precision city map information.

This paper attempts to build a full dimensional radio map from \ac{rss}
measurements between scattered low altitude aerial nodes and ground
nodes. The core idea is to reconstruct the geometry of the propagation
environment with \emph{radio semantics} embedded, such that one can
infer whether the propagation is under \ac{los}, slightly obstructed,
or in deep shadow, and exploit this information to predict the channel
gain between any two wireless nodes. Specifically, we build a multi-class
3D environment model, where the environment does \emph{not} necessarily
represent the visual appearance of buildings, but it is a model that
captures the \emph{semantic} meaning of how the signal strength will
be affected by the surrounding. For instance, a pillar made of plastic
may impose less attenuation on radio signals than a concrete pillar
may do. Mathematically, we construct\emph{ }multi-class 3D virtual
obstacles to describe the propagation environment for any pair of
wireless nodes in 3D, and based on this, we construct an environment-aware
radio map. We show that the proposed approach not only reconstructs
the geometry of the propagation environment, but also achieves better
accuracy in radio map construction than conventional interpolation-type
methods. Moreover, as the information is compressed in the proposed
multi-class 3D virtual obstacle model, it becomes easier to convey
and share radio maps in the network.

\subsection{Related Work}

\emph{Data-driven approaches:} Radio maps have been studied a lot
for indoor localization \cite{JiaMaLiuDou:J16,ZhaMa:J21}. In these
scenarios, the radio signatures are sampled through \ac{rss} measurement
over a 2D area, and the focus there was to handle sparsity and interpolate
the measurement data. Some of these methods include \ac{knn} \cite{NiNgu:C08,DenJiaZhoCui:C18},
sparse matrix or tensor processing \cite{ZhaMa:J21}, and Kriging
\cite{BraJemForMou:J16,SatFuj:J17}. Based on the recent advance of
image processing, deep learning for radio map construction was also
investigated in \cite{MasFarImr:C19,MasMarCheLi:J19,TegRom:J20,LevYapKutCai:J21,ShrFuHon:J22}.
Note that these approaches were mainly designed for 2D radio maps,
and they may not be easily extended to our scenario of interest.

\emph{Model-based approaches:} Conventional channel models first label
the local area into fine categories, such as urban and sub-urban,
and then select a parametric model from the fine category \cite{AlhKanJam:C14,YouZha:J20}.
From the view of signal propagation, some works \cite{FanZhoWei:J19,EsrGanGes:C21}
classify signals into \ac{los} and \ac{nlos} and then adopts path
loss model with a city map model. However, the limitation is that
the categorization is usually objective, and there are usually a limited
number of predefined models and sets of parameters to choose from.

\emph{Ray-tracing:} These methods are based on the 3D model of the
environment and compute the physical propagation paths through analyzing
any possible reflection, diffraction, and scattering \cite{EleElsYou:C11,SugSasOsaFur:J20}.
However, they are not only computationally expensive but also very
sensitive to the precision and the accuracy of the information available,
including the fine 3D model of the structure and its material.\textcolor{red}{}

\subsection{Our Contributions}

The paper aims at addressing the following two main issues:
\begin{itemize}
\item \emph{How to model the radio map with the geometry of radio environment
embedded;}
\item \emph{How to develop efficient algorithms to construct both the 3D
environment and the radio map.}
\end{itemize}

We develop a radio map model that consists of a multi-class 3D virtual
obstacle model to capture the geometry of the radio environment. The
key intuition is that if a link is relatively weak considering its
propagation distance, then there should be one or more obstacles that
block the direct path of the link. With the proposed model, the radio
map construction problem is transformed into a joint estimation problem
of the propagation parameters and the location and height of the 3D
virtual obstacles. A 3D city map is \emph{not} required, although
it can help better initialize the algorithm.

Our earlier work \cite{CheYatGes:C17} clusters the measurement into
different propagation conditions, such as \ac{los} and \ac{nlos},
and the follow-up work \cite{CheEsrGesMit:C17,ZhaChe:C20} designs
a preliminary virtual obstacle model based on the estimated \ac{los}
labels. However, the existing approach is an open-loop method, where
if the \ac{los} label obtained from \cite{CheYatGes:C17} were wrong,
the error would propagate and be magnified in the subsequent steps
in \cite{CheEsrGesMit:C17}. In this paper, we circumvent this limitation
by developing a new model and new algorithm to estimate the \emph{virtual
obstacle map} directly from the measurement data.

The novelty and contribution are summarized as follows:
\begin{itemize}
\item We develop a novel radio map model that consists of a parametric sub-model
that captures the geometry of the propagation environment with radio
semantics and a non-parametric sub-model that captures the residual
of the shadowing. With such a framework, the model performs well with
both small or large amount of training data. 
\item We formulate a least-squares estimation problem for the radio map
construction. While the problem is non-convex with degenerated gradient,
we discover and prove the \emph{partial} \emph{quasiconvexity} of
the problem; \lyxdeleted{User}{Thu Sep  1 09:54:13 2022}{ }based on
this theoretical result, we develop an efficient algorithm to construct
the radio map as well as the geometry of the virtual environment.
\item We conduct numerical experiments to verify that the proposed approaches
significantly outperform \ac{knn} and Kriging for radio map construction
using the city data of Shanghai. With the reconstruction of the geometry
of the virtual environment, we also demonstrate the performance advantage
of applying the proposed radio map model to UAV-assisted wireless
communication.
\end{itemize}

The rest of the paper is organized as follows. In Section \ref{sec:radio-map-model},
the multi-degree channel and multi-class virtual obstacle model are
established. Section \ref{sec:Radio-Map-Learning} and Section \ref{sec:Reconstruction-the-Shadowing}
develop the algorithm and establish theoretical results. Two applications
and their numerical results are demonstrated in Section \ref{sec:Applications-and-Numerical}.
Conclusions are drawn in Section \ref{sec:Conclusion}.

\section{Radio Map Model\label{sec:radio-map-model}}

Consider wireless communications between a ground user and a low altitude
aerial node over a dense urban environment. Typically, the aerial
node can be a relay \ac{bs} carried by a \ac{uav} or a sensing device
installed on a high tower or on the rooftop. The focus of this paper
is to construct a radio map to characterize the channel gain between
every terrestrial (ground) user position $\mathbf{p}_{\text{u}}$
and every aerial node (drone) position $\mathbf{p}_{\text{d}}$ pair
based on a limited number of measurement samples, where $\mathbf{p}_{\text{u}},\mathbf{p}_{\text{d}}\in\mathbb{R}^{3}$.

\subsection{Radio Map with Environment Semantics}

Denote the \emph{communication link} $\mathbf{p}=(\mathbf{p}_{\text{u}},\mathbf{p}_{\text{d}})\in\mathbb{R}^{6}$
as the positions of the ground user and aerial node pair. The signal
that propagates between $\mathbf{p}_{\text{u}}$ and $\mathbf{p}_{\text{d}}$
consists of multiple paths which experience penetration, reflection,
diffraction, and scattering according to the specific environment.
We aim at building a radio map $g(\mathbf{p})$ to capture the large-scale
effect of the channel gain, including the path loss and the shadowing
for each link characterized by the $6$-dimensional location $\mathbf{p}$.

Recall a classical channel model $\sum_{k=0}^{1}\big(\beta_{k}+\alpha_{k}\log_{10}\|\mathbf{p}_{\text{u}}-\mathbf{p}_{\text{d}}\|_{2}\big)\mathbb{I}\{\mathbf{p}\in\mathcal{D}_{k}\}+\xi$,
which is based on the $0$-$1$ categorization of links being in
the \ac{los} region $\mathbf{p}\in\mathcal{D}_{0}$ or in \ac{nlos}
region $\mathbf{p}\in\mathcal{D}_{1}$; and $\xi$ is a random variable
for the shadowing. We extend such a classical model to a $(K+1)$-degree
model with soft categorization. Specifically, the proposed radio map
model $g(\mathbf{p})$ consists of a \emph{deterministic radio map}
$\bar{g}(\mathbf{p};\bm{\theta},\mathbf{H})$ and a \emph{residual
shadowing map} $\xi(\mathbf{p})$: 
\begin{equation}
g(\mathbf{p})=\bar{g}(\mathbf{p};\bm{\theta},\mathbf{H})+\xi(\mathbf{p})\label{eq:channel-model}
\end{equation}
where the deterministic radio map
\begin{equation}
\bar{g}(\mathbf{p};\bm{\theta},\mathbf{H})\triangleq\sum_{k=0}^{K}\big(\beta_{k}+\alpha_{k}\log_{10}\|\mathbf{p}_{\text{u}}-\mathbf{p}_{\text{d}}\|_{2}\big)S_{k}(\mathbf{p};\mathbf{H})\label{eq:channel-model-nonoi}
\end{equation}
is parameterized by $\bm{\theta}=\{\alpha_{k},\beta_{k}\}{}_{k=0}^{K}$
for each path loss sub-model $\beta_{k}+\alpha_{k}\log_{10}\|\mathbf{p}_{\text{u}}-\mathbf{p}_{\text{d}}\|_{2}$
under different degrees of signal obstruction; the term $S_{k}(\mathbf{p};\mathbf{H})$
models the likelihood that link $\mathbf{p}$ experiences in the $k$th
degree of signal obstruction with parameter $\mathbf{H}$, to be explained
in the next subsection, to capture the semantic information of the
propagation environment. The component $\xi(\mathbf{p})$ is a random
process assumed with zero mean and bounded variance.

The advantage of the proposed model (\ref{eq:channel-model}) is as
follows: First, as suggested by measurement data that there are rarely
sharp edges between \ac{los} and \ac{nlos}, a probabilistic function
$S_{k}(\mathbf{p};\mathbf{H})$ approximates the reality better than
an indicator function does in a classical model. Second, intuitively,
the more segments $\mathcal{D}_{k}$ to estimate, the more precise
the model $\bar{g}(\mathbf{p};\bm{\theta},\mathbf{H})$ may approximate
the reality, resulting in a lower variance for the random component
$\xi(\mathbf{p})$. Third, the geometry of the radio propagation environment
will be explicitly built into the model $\bar{g}(\mathbf{p};\bm{\theta},\mathbf{H})$
to assist radio map reconstruction, as to be shown later.

\subsection{Virtual Obstacle Model\label{subsec:Multi-class-Virtual-Obstacle}}

We propose to impose a \emph{multi-class virtual obstacle model} for
$S_{k}(\mathbf{p};\mathbf{H})$, the likelihood of $\mathbf{p}$ being
in the $k$th propagation region $\mathcal{D}_{k}$ to describe the
environment semantics. The general idea is to employ an \emph{equivalent}
virtual obstacle at a certain \emph{location}, with appropriate \emph{height}
and \emph{type} to intersect with the direct path $\mathbf{p}$ to
represent the likelihood $S_{k}(\mathbf{p};\mathbf{H})$. Hence, the
virtual obstacle may not be mapped to a building in the reality, but
serves as a geometry representation of the radio propagation environment.

For example, as illustrated in Fig. \ref{fig:sigray2clas}, if a link
$\mathbf{p}$ is in deep shadow, then we place a solid virtual obstacle
to intersect the direct path of $\mathbf{p}$; on the other hand,
if $\mathbf{p}$ is in light shadow, then some light virtual obstacle
is in place to intersect the direct path.
\begin{figure}
\begin{centering}
\includegraphics[width=0.5\columnwidth]{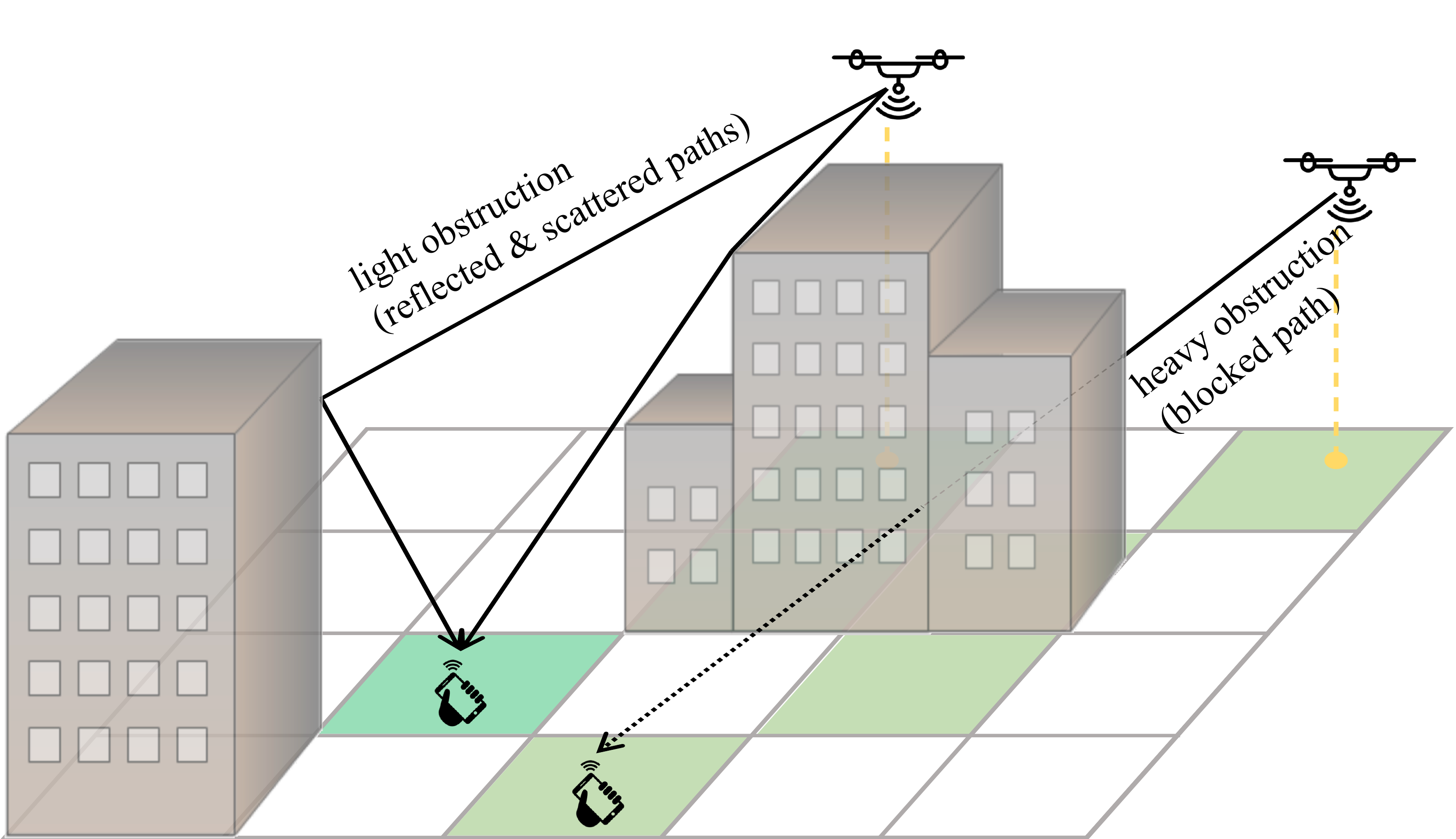}\includegraphics[width=0.5\columnwidth]{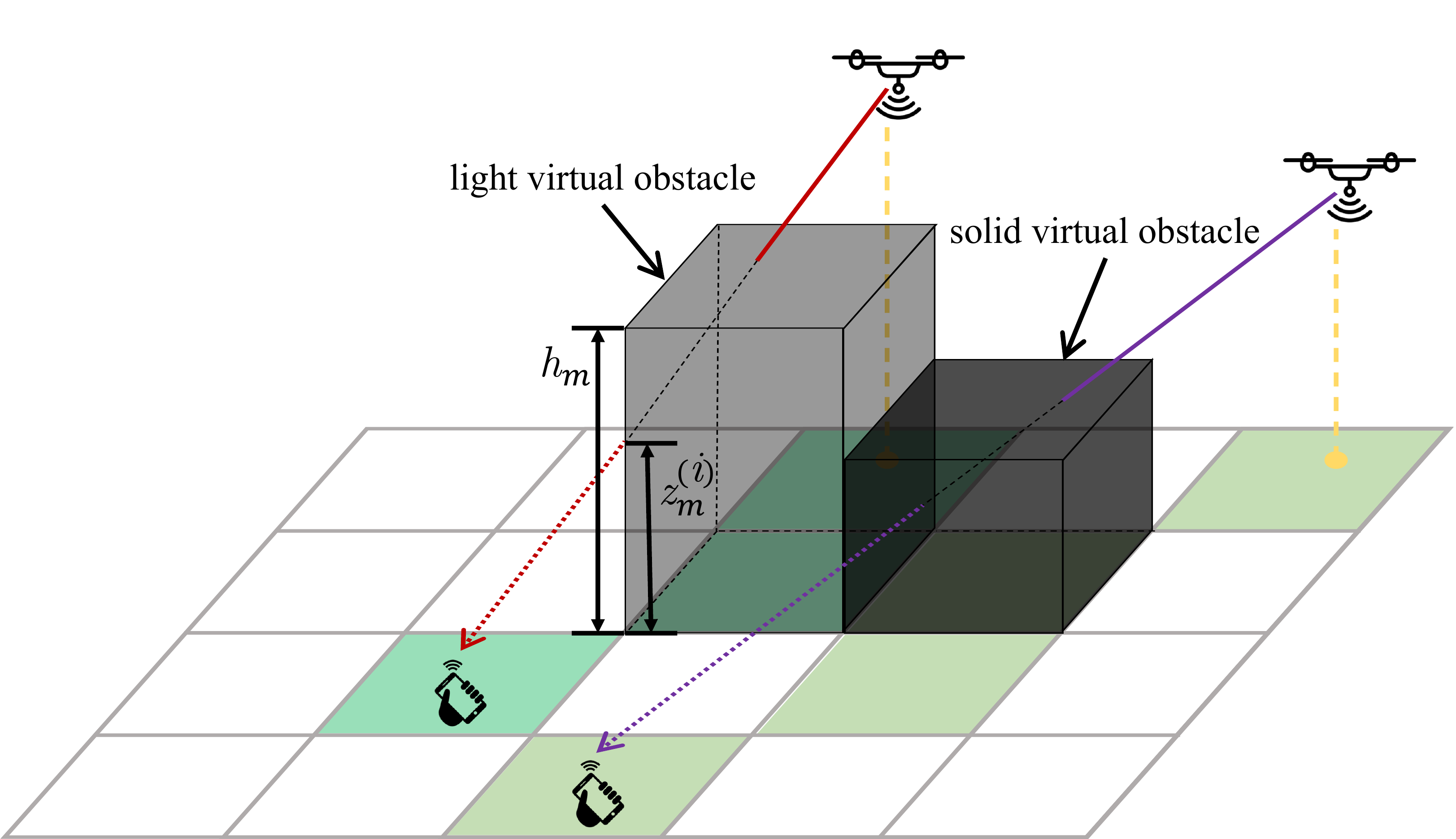}
\par\end{centering}
\caption{Left: the multi-path propagation in reality; right: a virtual obstacle
structure that captures the semantic information of the propagation
environment.}
\label{fig:sigray2clas}
\end{figure}

Specifically, we first partition the target ground area into $M$
grid cells. A common partition approach is to use the uniform square
grid with appropriate spacing according to the resolution requirement
and the amount of measurement data. Then, a virtual obstacle map can
be represented by an $M\times K$ matrix $\mathbf{H}$, where the
$k$th column $\mathbf{h}_{k}$ of $\mathbf{H}$ represents the height
of class-$k$ virtual obstacle, and the $m$th entry of $\mathbf{h}_{k}$
represents the height of the virtual obstacle located at the $m$th
grid cell. With such a notation, the likelihood $S_{k}(\mathbf{p};\mathbf{H})$
of link $\mathbf{p}$ belonging to propagation region $\mathcal{D}_{k}$
can be constructed from the virtual obstacle map $\mathbf{H}$ as
follows.

\subsubsection{Propagation Regions with Hard Boundary}

It holds that $\sum_{k}S_{k}(\mathbf{p};\mathbf{H})=1$ and $S_{k}(\mathbf{p};\mathbf{H})\in\{0,1\}$,
and constructed based on the following rule:
\begin{itemize}
\item $\mathbf{p}\in\mathcal{D}_{0}$, \emph{i.e.}, $S_{0}(\mathbf{p};\mathbf{H})=1$,
if there is no obstacle that intersects with the direct path between
$\mathbf{p}_{\text{u}}$ and $\mathbf{p}_{\text{d}}$.
\item $\mathbf{p}\in\mathcal{D}_{k}$, \emph{i.e.}, $S_{k}(\mathbf{p};\mathbf{H})=1$
for some $k\ge1$, if a class-$k$ obstacle intersects with the direct
path between $\mathbf{p}_{\text{u}}$ and $\mathbf{p}_{\text{d}}$,
while no class-$l$, $l>k$, obstacle intersecting with the direct
path.
\end{itemize}

Mathematically, denote $\mathscr{\mathscr{\mathcal{B}}}^{(i)}$ as
the set of grid cells that are covered by the direct path between
$\mathbf{p}_{\text{u}}^{(i)}$ and $\mathbf{p}_{\text{d}}^{(i)}$.
In other words, if one projects the path $(\mathbf{p}_{\text{u}}^{(i)},\mathbf{p}_{\text{d}}^{(i)})$
onto the ground, then the projected path passes through the grid cells
and only the grid cells in $\mathcal{B}^{(i)}$. For each grid cell
$m\in\mathscr{\mathscr{\mathcal{B}}}^{(i)}$, denote $z_{m}^{(i)}$
as the altitude when the path $(\mathbf{p}_{\text{u}}^{(i)},\mathbf{p}_{\text{d}}^{(i)})$
passes over the grid cell.

It follows from the first rule above that $\mathbf{p}^{(i)}\in\mathcal{D}_{0}$
if $h_{m,k}<z_{m}^{(i)}$ for \emph{all} $m\in\mathcal{B}^{(i)}$
and \emph{all} $0<k\leq K$, \emph{i.e.}, for all the relevant grid
locations $m\in\mathcal{B}^{(i)}$, all classes of obstacles are below
the corresponding critical altitude $z_{m}^{(i)}$; mathematically,
$\mathbb{I}\{\mathbf{p}^{(i)}\in\mathcal{D}_{0}(\mathbf{H})\}=\mathbb{I}\{h_{m,l}<z_{m}^{(i)},\forall m\in\mathscr{\mathscr{\mathcal{B}}}^{(i)},\forall l>0\}$.
From the second rule above, if $\mathbf{p}^{(i)}\in\mathcal{D}_{k}$
for $k>0$, we must have $h_{m,k}\ge z_{m}^{(i)}$ for \emph{some}
class-$k$ obstacle at the $m$th grid cell, mathematically, $\mathbb{I}\{h_{m,k}\ge z_{m}^{(i)},\exists m\in\mathscr{\mathscr{\mathcal{B}}}^{(i)}\}=1$,
and at the same time, we should also have $h_{m,l}<z_{m}^{(i)}$ for
\emph{all} $m\in\mathcal{B}^{(i)}$ and \emph{all} $l\geq k+1$, mathematically,
$\mathbb{I}\{h_{m,l}<z_{m}^{(i)},\forall m\in\mathscr{\mathscr{\mathcal{B}}}^{(i)},\forall l>k\}=1$.
To summarize, it follows that
\begin{equation}
\mathbb{I}\{\mathbf{p}^{(i)}\in\mathcal{D}_{k}(\mathbf{H})\}=\begin{cases}
\underset{m\in\mathscr{\mathscr{\mathcal{B}}}^{(i)}}{\prod}\underset{l>0}{\prod}\mathbb{I}\{h_{m,l}<z_{m}^{(i)}\}, & \textrm{if }k=0,\\
\\
\begin{array}{l}
\big(1-\underset{m\in\mathscr{\mathscr{\mathcal{B}}}^{(i)}}{\prod}(1-\mathbb{I}\{h_{m,k}\ge z_{m}^{(i)}\})\big)\quad\\
\hfill\times\underset{m\in\mathscr{\mathscr{\mathcal{B}}}^{(i)}}{\prod}\underset{l>k}{\prod}\mathbb{I}\{h_{m,l}<z_{m}^{(i)}\},
\end{array} & \textrm{if }k\ge1.
\end{cases}\label{eq:esthidvircon}
\end{equation}
Using the propagation regions defined in (\ref{eq:esthidvircon}),
the likelihood function $S_{k}(\mathbf{p};\mathbf{H})$ in (\ref{eq:channel-model})
can be chosen as $S_{k}(\mathbf{p};\mathbf{H})\triangleq\mathbb{I}\{\mathbf{p}\in\mathcal{D}_{k}(\mathbf{H})\}$.

\subsubsection{Propagation Regions with Soft Boundary}

To allow the likelihood $S_{k}(\mathbf{p};\mathbf{H})$ to take fractional
numbers in $[0,1]$, we extend the hard boundary model by applying
a spatial low-pass filter on the propagation regions $\mathcal{D}_{k}$
defined in (\ref{eq:esthidvircon}). Specifically, given a link $\mathbf{p}$,
we evaluate a set of neighbor positions with offset $\bm{\epsilon}_{j}$
from $\mathbf{p}$. By averaging $\mathbb{I}\{\mathbf{p}+\bm{\epsilon}_{j}\in\mathcal{D}_{k}(\mathbf{H})\}$
with weights $\omega_{j}$, one can obtain the likelihood $S_{k}(\mathbf{p};\mathbf{H})$.
A common choice of filter coefficients $\omega_{j}$ can be obtained
as a function of the distance $\left\Vert \bm{\epsilon}_{j}\right\Vert _{2}$
from $\mathbf{p}$. In this paper, we consider a spatial filter that
consists of a set of $\ensuremath{J}$ uniform grid points $\bm{\epsilon}_{j},j=0,1,\dots,J-1,$
in 6D space centered at the origin with $\bm{\epsilon}_{0}$ chosen
as $\bm{\epsilon}_{0}=\mathbf{0}$, and the weights are chosen as
$\omega_{j}=c\cdot\textrm{exp}(-\left\Vert \bm{\epsilon}_{j}\right\Vert _{2}^{2}/\sigma_{\omega}^{2})$,
where $\sigma_{\omega}$ is a parameter, and $c$ is a normalization
factor such that $\sum_{j=0}^{J-1}\omega_{j}=1$.\footnote{In general, the weights can be designed using a kernel function, where
the smaller $\|\bm{\epsilon}_{j}\|_{2}$, the larger the weight. The
choice of the kernel and its parameters can be determined using a
cross-validation approach.} Thus, the likelihood function $S_{k}(\mathbf{p};\mathbf{H})$ parameterized
by the virtual obstacle map $\mathbf{H}$ is defined as
\begin{equation}
S_{k}(\mathbf{p};\mathbf{H})=\sum_{j=0}^{J-1}\omega_{j}\mathbb{I}\{\mathbf{p}+\bm{\epsilon}_{j}\in\mathcal{D}_{k}(\mathbf{H})\}\label{eq:aveinddis}
\end{equation}
which satisfies $\sum_{k=0}^{K}S_{k}(\mathbf{p};\mathbf{H})=1$.

\section{Radio Map Construction via Environment Mapping\label{sec:Radio-Map-Learning}}

In this section, we jointly estimate the propagation parameter $\bm{\theta}$
and the virtual obstacle map $\mathbf{H}$ for constructing the deterministic
radio map $\bar{g}(\mathbf{p};\bm{\theta},\mathbf{H})$ in (\ref{eq:channel-model-nonoi}).

\subsection{Formulation of the Radio Map Learning Problem\label{subsec:The-Measurement-Model}}

Consider taking measurements at transmit and receive location pairs
$\{\mathbf{p}^{(i)}\}$ and recall $\mathbf{p}=(\mathbf{p}_{\text{u}},\mathbf{p}_{\text{d}})$.
Based on (\ref{eq:channel-model}), the measured \ac{rss} can be
written as 
\begin{equation}
y^{(i)}=\bar{g}(\mathbf{p}^{(i)};\bm{\theta},\mathbf{H})+n^{(i)}\label{eq:recsigmea}
\end{equation}
where $n^{(i)}=\xi(\mathbf{p}^{(i)})+\tilde{n}^{(i)}$ captures both
the random component $\xi(\mathbf{p})$ in (\ref{eq:channel-model})
for the residual shadowing and the measurement noise $\tilde{n}^{(i)}$
which is assumed as independent and identically distributed with zero
mean, variance $\sigma_{\text{n}}^{2}$, and finite fourth-order moment.

The goal of this section is to estimate parameters $\bm{\theta}$
and $\mathbf{H}$ from the set of noisy measurement data $\{(\mathbf{p}^{(i)},y^{(i)})\}_{i=1}^{N}$
obtained from (\ref{eq:recsigmea}). A least-squares problem can be
formulated as follows
\begin{equation}
\mathop{\textrm{minimize}}\limits _{\bm{\theta},\mathbf{H}\succeq\mathbf{0}}\quad f(\bm{\theta},\mathbf{H})\triangleq\frac{1}{N}\sum\limits _{i=1}^{N}\Big[y^{(i)}-\sum_{k=0}^{K}\big(\beta_{k}+\alpha_{k}d(\mathbf{p}^{(i)})\big)S_{k}(\mathbf{p}^{(i)};\mathbf{H})\Big]^{2}\label{eq:formulation}
\end{equation}
where $d(\mathbf{p}^{(i)})\triangleq\log_{10}\|\mathbf{p}_{\text{u}}^{(i)}-\mathbf{p}_{\text{d}}^{(i)}\|_{2}$
is the log-distance between $\mathbf{p}_{\text{u}}^{(i)}$ and $\mathbf{p}_{\text{d}}^{(i)}$
for the $i$th measurement.

Note that the least-squares problem (\ref{eq:formulation}) is difficult
to solve using a standard solver. This is because the problem is \emph{non-convex}
in the joint variable $(\bm{\theta},\mathbf{H})$ and the objective
function $f(\bm{\theta},\mathbf{H})$ is \emph{discontinuous} due
to the indicator functions used in (\ref{eq:channel-model-nonoi})--(\ref{eq:aveinddis}).
To circumvent these difficulties, we will exploit the property discovered
in $f(\bm{\theta},\mathbf{H})$.

\subsection{Asymptotic Consistency of Radio Maps}

Since the objective function $f(\bm{\theta},\mathbf{H})$ in (\ref{eq:formulation})
contains randomness due to the measurement noise, we first find a
\emph{deterministic proxy} for $f(\bm{\theta},\mathbf{H})$ under
large $N$. Denote $\bm{\theta}^{*}$ and $\mathbf{H}^{*}$ as the
true parameters in the measurement model (\ref{eq:recsigmea}) and
consider the following deterministic proxy function
\begin{equation}
\bar{f}(\bm{\theta},\mathbf{H})\triangleq\frac{1}{N}\sum\limits _{i=1}^{N}\Big[\bar{g}(\mathbf{p}^{(i)};\bm{\theta},\mathbf{H})-\bar{g}(\mathbf{p}^{(i)};\bm{\theta}^{*},\mathbf{H}^{*})\Big]^{2}\label{eq:formulnonoi}
\end{equation}
and it is clear that $\bar{f}(\bm{\theta}^{*},\mathbf{H}^{*})=0$.
\begin{lem}[Deterministic Equivalence]
Suppose that the random component $\xi(\mathbf{p}^{(i)})$ in the
measurement model (\ref{eq:recsigmea}) is weakly dependent, i.e.,
the covariance satisfies $\mbox{cov}(\xi(\mathbf{p}^{(i)}),\xi(\mathbf{p}^{(j)}))\to0$,
as $|i-j|\to\infty$ and $N\to\infty$, and moreover, the limit $\lim_{N\to\infty}\frac{1}{N}\sum_{i=1}^{N}\xi(\mathbf{p}^{(i)})^{2}$
exists and is finite.\footnote{If the process $\xi(\mathbf{p})$ is segment-wise second-order stationary
within each propagation segment, then the limit exists if one samples
each propagation segment with a fixed probability, for example, under
uniform sampling over the entire area.} Then, there exists a finite constant $0<C<\infty$, such that
\[
f(\bm{\theta},\mathbf{H})\to\bar{f}(\bm{\theta},\mathbf{H})+C
\]
in probability, for every $(\bm{\theta},\mathbf{H})$, as $N\to\infty$.\label{lem:asyequ}
\end{lem}
\begin{proof}
See Appendix \ref{app:prooflemnoi}.
\end{proof}
It follows from the above lemma that the parameters $(\bm{\theta},\mathbf{H})$
that minimize the deterministic proxy $\bar{f}(\bm{\theta},\mathbf{H})$
also minimize the least-squares cost $f(\bm{\theta},\mathbf{H})$
asymptotically. As a result, analyzing the property of the deterministic
proxy $\bar{f}(\bm{\theta},\mathbf{H})$ may inspire efficient algorithms
to solve (\ref{eq:formulation}) under large N.

It is also observed that there could be multiple local minima for
both $f(\bm{\theta},\mathbf{H})$ and $\bar{f}(\bm{\theta},\mathbf{H})$.
Specifically, the globally optimal solution $(\hat{\bm{\theta}},\hat{\mathbf{H}})$
to the least-squares problem (\ref{eq:formulation}) may differ from
the true parameter $(\bm{\theta}^{*},\mathbf{H}^{*})$ even at the
asymptotic regime. Yet, we are not interested in the estimated parameters
$(\hat{\bm{\theta}},\hat{\mathbf{H}})$, but the radio map $\bar{g}(\mathbf{p};\hat{\bm{\theta}},\hat{\mathbf{H}})$
constructed from these parameters. In other words, the focus is not
on reconstructing the actual buildings or identifying the true parameters
$(\bm{\theta}^{*},\mathbf{H}^{*})$, but to extract consistent environment
semantics for constructing radio maps.

The following corollary shows that one may obtain \emph{asymptotically
consistent} radio maps even the globally optimal solution $(\hat{\bm{\theta}},\hat{\mathbf{H}})$
differs from the true parameter $(\bm{\theta}^{*},\mathbf{H}^{*})$;
here, consistency means that radio maps are identical $\bar{g}(\mathbf{p}^{(i)};\hat{\bm{\theta}},\hat{\mathbf{H}})=\bar{g}(\mathbf{p}^{(i)};\bm{\theta}^{*},\mathbf{H}^{*})$
at the measurement locations $\{\mathbf{p}^{(i)}\}$.
\begin{cor}[Asymptotic Consistency of Radio Maps]
It holds that, as $N\to\infty$,
\[
\bar{f}(\hat{\bm{\theta}},\hat{\mathbf{H}})=\frac{1}{N}\sum_{i=1}^{N}\left(\bar{g}(\mathbf{p}^{(i)};\hat{\bm{\theta}},\hat{\mathbf{H}})-\bar{g}(\mathbf{p}^{(i)};\bm{\theta}^{*},\mathbf{H}^{*})\right)^{2}\to0
\]
where $(\hat{\bm{\theta}},\hat{\mathbf{H}})$ is the globally optimal
solution to (\ref{eq:formulation}).\label{cor:asyequrad}
\end{cor}
\begin{proof}
By definition, the minimum value of $\bar{f}(\bm{\theta},\mathbf{H})$
is obtained as $\bar{f}(\bm{\theta}^{*},\mathbf{H}^{*})=0$, and therefore,
$\bar{f}(\bm{\theta},\mathbf{H})+C$ can be globally minimized to
$C$. Since $f(\bm{\theta},\mathbf{H})\to\bar{f}(\bm{\theta},\mathbf{H})+C$
from Lemma \ref{lem:asyequ}, $f(\bm{\theta},\mathbf{H})$ is also
asymptotically and globally minimized to $f(\hat{\bm{\theta}},\hat{\mathbf{H}})\to C$,
which implies that $\bar{f}(\hat{\bm{\theta}},\hat{\mathbf{H}})\to0$
as $N\to\infty$.
\end{proof}

\subsection{Solution to the Propagation Parameter $\bm{\theta}$\label{subsec:Reformulation-for-an}}

It can be easily verified that given the variable $\mathbf{H}$, the
problem (\ref{eq:formulation}) is convex in $\bm{\theta}$. To see
this, denote $\mathbf{X}\in\mathbb{R}^{N\times(2K+2)}$ as log-distance
data matrix, where the even elements in the $i$th row of $\mathbf{X}$
equal to $d(\mathbf{p}^{(i)})$ and the odd elements in the $i$th
row of $\mathbf{X}$ equal to $1$. Arrange the elements in the variable
$\bm{\theta}\in\mathbb{R}^{2K+2}$ as $\bm{\theta}=[\alpha_{0}\;\beta_{0}\;\alpha_{1}\;\beta_{1}\;\cdots\;\alpha_{K}\;\beta_{K}]^{\textrm{T}}$
as the path loss parameter vector for the sub-models. Stack the measurement
value $y^{(i)}$ into a vector $\mathbf{y}=[y^{(1)}\;y^{(2)}\;\cdots\;y^{(N)}]^{\textrm{T}}\in\mathbb{R}^{N}$.
Finally, denote $\mathbf{S}\in\mathbb{R}^{N\times(2K+2)}$ as the
likelihood matrix, where $[\mathbf{S}]_{i,2k}=[\mathbf{S}]_{i,2k+1}=S_{k}(\mathbf{p}^{(i)};\mathbf{H})$.
Then, for a fixed $\mathbf{H}$, problem (\ref{eq:formulation}) can
be written as

\begin{equation}
\underset{\bm{\theta}}{\mathop{\textrm{minimize}}}\quad\left\Vert \mathbf{(S\circ X)\bm{\theta}-y}\right\Vert _{2}^{2}\label{eq:wlsall}
\end{equation}
where $\circ$ is the Hadamard product, \emph{i.e}., $[\mathbf{S\circ X}]_{ij}=[\mathbf{S}]_{ij}[\mathbf{X}]_{ij}$.
Problem (\ref{eq:wlsall}) is unconstrained quadratic programming,
and it is convex \ac{wrt} $\bm{\theta}$. It can be solved by setting
the derivative to zero, and the solution is given by
\begin{equation}
\hat{\bm{\theta}}=\big((\mathbf{S\circ X})^{\textrm{T}}(\mathbf{S\circ X})\big)^{-1}(\mathbf{S\circ X})^{\textrm{T}}\mathbf{y}.\label{eq:cloforsolthe}
\end{equation}

\begin{prop}
Under $\mathbf{H}^{*}$, the solution in (\ref{eq:cloforsolthe})
is an unbiased estimator of $\bm{\theta}^{*}$, \emph{i.e}., $\mathbb{E}\{\hat{\bm{\theta}}\}=\bm{\theta}^{*}$.
\end{prop}
\begin{proof}
With $\mathbf{H}^{*}$, the observation model can be written as $\mathbf{y}=(\mathbf{S}^{*}\circ\mathbf{X})\bm{\theta}+\mathbf{n}$,
where $\mathbf{n}$ is a vector stacking $n$. Therefore, it is a
standard least-squares estimation problem for a linear observation
model $\mathbf{y}$ with zero mean noise $\mathbf{n}$. It is well-known
that the least-squares estimator (\ref{eq:cloforsolthe}) in this
case is unbiased \cite{kay1993fundamentals}.
\end{proof}

\subsection{Quasiconvexity in the Environment Parameter $\mathbf{H}$}

While problem (\ref{eq:formulation}) is convex in $\bm{\theta}$
by fixing $\mathbf{H}$, it is still non-convex in $\mathbf{H}$ by
fixing $\bm{\theta}$. However, we discover that $\mathbf{H}$ is
partially quasiconvex, which can be later exploited for efficient
algorithm design.

First, consider the $K=1$ case, where there are two propagation regions,
\ac{los} and \ac{nlos}, and the matrix $\mathbf{H}$ degenerates
to a column vector $\mathbf{h}$.
\begin{thm}[Quasiconvexity for $K=1$ under Soft Boundary]
Suppose that the filter coefficient $\omega_{0}$ in (\ref{eq:aveinddis})
satisfies $\omega_{0}\ge\frac{2}{3}$. Given a vector $\mathbf{h}'\succeq\mathbf{h}^{*}$
and an index $m$, consider the interval $\mathcal{I}_{m}\triangleq\{\mathbf{h}\in\mathbb{R}^{M}:0\leq h_{m}\leq H_{\textrm{max}},h_{j}=h'_{j},\forall j\neq m\}$.
Then, $\bar{f}(\bm{\theta}^{*},\mathbf{h})$ in (\ref{eq:formulnonoi})
is quasiconvex over the interval $\mathcal{I}_{m}$.\label{thm:quatwosof}
\end{thm}
\begin{proof}
See Appendix \ref{app:proofthebas}.
\end{proof}
The above result implies that given a variable $\mathbf{h}\succeq\mathbf{h}^{*}$,
the function $\bar{f}(\bm{\theta}^{*},\mathbf{h})$ is partially quasiconvex
\ac{wrt} to each entry $h_{m}$ with all the other entries $h_{j}$,
$j\neq m$, fixed. As a result of the partial quasiconvexity, there
exists $\hat{h}_{m}$, such that, for $h_{m}<\hat{h}_{m}$, $\bar{f}$
is non-increasing in $h_{m}$, and for $h_{m}>\hat{h}_{m}$, $\bar{f}$
is non-decreasing.

Next, consider the case of a general $K$ and the propagation regions
being modeled with hard boundaries, \emph{i.e}., the likelihood function
is chosen as $S_{k}(\mathbf{p};\mathbf{H})=\mathbb{I}\{\mathbf{p}\in\mathcal{D}_{k}(\mathbf{H})\}$
as in (\ref{eq:esthidvircon}). The following result shows that the
partial quasiconvexity in Theorem \ref{thm:quatwosof} also holds.
\begin{thm}[Quasiconvexity under Hard Boundary]
Given a matrix $\mathbf{H}'\succeq\mathbf{H}^{*}$ and an index $(m,k)$,
define the interval $\mathcal{I}_{m,k}\triangleq\{\mathbf{H}\in\mathbb{R}^{M\times K}:0\leq h_{m,k}\leq H_{\textrm{max}},h_{j,l}=h'_{j,l},\forall(j,l)\neq(m,k)\}$.
Then, $\bar{f}(\bm{\theta}^{*},\mathbf{H})$ is quasiconvex over the
interval $\mathcal{I}_{m,k}$.\label{thm:quaharbou}
\end{thm}
\begin{proof}
See Appendix \ref{app:proofthequa}.
\end{proof}
Theorems \ref{thm:quatwosof} and \ref{thm:quaharbou} imply that
if we focus on each individual entry $h_{m,k}$ in the variable $\mathbf{H}$,
then $\bar{f}(\bm{\theta},\mathbf{H})$ first decreases and then increases.
More specifically, according to the fact that $\bar{f}(\bm{\theta},\mathbf{H})$
in (\ref{eq:formulnonoi}) is constructed from a number of indicator
functions in (\ref{eq:esthidvircon}) and (\ref{eq:aveinddis}), $\bar{f}(\bm{\theta},\mathbf{H})$
appears like a \emph{staircase} function that first steps down along
the interval $\mathcal{I}_{m,k}$, reaching the bottom around $h_{m,k}^{*}$,
and then steps up, where the bottom appears as a flat \emph{basin}
as shown in Fig. \ref{fig:fmkftilde}.

To characterize the basin of $\bar{f}$ over $\mathcal{I}_{m,k}$,
define a function 
\begin{equation}
\bar{f}_{m,k}(h_{m,k};\bm{\theta},\mathbf{H}_{m,k}^{-})=\bar{f}(\bm{\theta},\mathbf{H})\label{eq:bar-f_mk}
\end{equation}
of the scalar variable $h_{m,k}$ with the other variables $\bm{\theta}$
and $\mathbf{H}_{m,k}^{-}$ held fixed, where $\mathbf{H}_{m,k}^{-}=\{h_{j,l}:\forall(j,l)\neq(m,k)\}$
is a collection of entries from the matrix $\mathbf{H}$ except the
$(m,k)$th one. The basin is defined as the interval $\mathcal{\underline{I}}_{m,k}(\bm{\theta},\mathbf{H})\triangleq\big\{ z:\bar{f}_{m,k}(z;\bm{\theta},\mathbf{H}_{m,k}^{-})\leq\bar{f}_{m,k}(h;\bm{\theta},\mathbf{H}_{m,k}^{-}),\forall0\leq h\leq H_{\text{max}}\big\}$.
We are interested in the largest value in the basin 
\begin{equation}
\hat{h}_{m,k}(\bm{\theta},\mathbf{H})\triangleq\sup\{\mathcal{\underline{I}}_{m,k}(\bm{\theta},\mathbf{H})\}\label{eq:hat_h_mk}
\end{equation}
where the algorithm to solve (\ref{eq:hat_h_mk}) will be developed
in Section \ref{subsec:Optimizing--via}. With the notation of $\hat{h}_{m,k}(\bm{\theta},\mathbf{H})$,
the following property can be established.
\begin{thm}[Consistency]
Suppose $\mathbf{H}\succeq\mathbf{H}^{*}$ and consider the interval
$\mathcal{I}_{m,k}$ as defined in Theorem \ref{thm:quaharbou}. Then,
$\hat{h}_{m,k}(\bm{\theta}^{*},\mathbf{H})\geq h_{m,k}^{*}$. Moreover,
given $\mathbf{H}''\succeq\mathbf{H}'\succeq\mathbf{H}^{*}$, it holds
that $\hat{h}_{m,k}(\bm{\theta}^{*},\mathbf{H}'')\geq\hat{h}_{m,k}(\bm{\theta}^{*},\mathbf{H}')\ge h_{m,k}^{*}$.\label{thm:dompro}
\end{thm}
\begin{proof}
See Appendix \ref{app:proofthedom}.
\end{proof}
As inspired from Theorems \ref{thm:quatwosof}\textendash \ref{thm:dompro},
when $\bm{\theta}$ is sufficiently close to $\bm{\theta}^{*}$, an
efficient algorithm to optimize $\mathbf{H}$ can proceed as follows.
First, set the initial value of $\mathbf{H}$ as $h_{m,k}=H_{\textrm{max}}$
for all $m,k$. Then, for each element $h_{m,k}$, find $\hat{h}_{m,k}$
that minimizes $\bar{f}(\bm{\theta},\mathbf{H})$, and repeat this
step until convergence. This approach is summarized in Algorithm \ref{alg:asycoodes}
and its convergence can be analyzed as follows.

Let $\mathbf{H}(t)$ denotes the variable $\mathbf{H}$ at the $t$th
iteration. Theorem \ref{thm:dompro} implies the convergence of $h_{m,k}(t)$
for $\bm{\theta}=\bm{\theta}^{*}$. Suppose at $t>1$, given that
$\mathbf{H}(t-1)\succeq\mathbf{H}(t)\succeq\mathbf{H}^{*}$, Theorem
\ref{thm:dompro} implies that $\hat{h}_{m,k}(t-1)\ge\hat{h}_{m,k}(t)\ge h_{m,k}^{*}$,
$\forall m,k$, under $\bm{\theta}^{*}$, and consequently, $\mathbf{H}(t)\succeq\mathbf{H}(t+1)\succeq\mathbf{H}^{*}$.
As $t=1$ can be checked to satisfy $\hat{h}_{m,k}(1)\ge\hat{h}_{m,k}(2)\ge h_{m,k}^{*}$
due to the initialization, by induction, $\hat{h}_{m,k}(t)\ge\hat{h}_{m,k}(t+1)\ge h_{m,k}^{*}$
is satisfied for all $t$, which means that Algorithm \ref{alg:asycoodes}
constructs monotonically decreasing and lower bounded sequences $h_{m,k}(t)$,
which implies the convergence of $\mathbf{H}(t)$.

As a side note, finding $\bm{\theta}^{*}$ is relatively easier as
suggested by a lot of numerical experiments. The intuition is that
$\bm{\theta}^{*}$ has just $2K+2$ variables, and thus, depends on
global statistics, not very sensitive to $\mathbf{H}$ under large
$N$.

\begin{algorithm}
\begin{enumerate}
\item Initialize $\mathbf{H}(1)=\mathbf{1}H_{\textrm{max}}$, $\bm{\theta}(1)$
using \ac{em} in \cite{CheYatGes:C17}, and iteration $t=1$.
\item Optimize $\mathbf{H}$: For each $(m,k)$, update $h_{m,k}(t+1)=\hat{h}_{m,k}$
as defined in (\ref{eq:hat_h_mk}) based on $\bm{\theta}(t)$ and
$\mathbf{H}_{m,k}^{-}(t)$. Specifically, the following bisection
search is used:\label{enu:forasyupd}
\begin{enumerate}
\item Initialize $h_{\min}=0$ and $h_{\max}=H_{max}$.
\item Set $h_{m,k}=\frac{1}{2}(h_{\min}+h_{\max})$ and find the minimizer
$a_{1}^{*}$ of (\ref{eq:locpolgra}). \label{enu:forasybis}
\item If $a_{1}^{*}<0$, then $h_{\min}\leftarrow h_{m,k}$; if $a_{1}^{*}>0$,
then $h_{\max}\leftarrow h_{m,k}$.
\item Repeat from Step \ref{enu:forasybis} until $|h_{\max}-h_{\min}|<\epsilon$,
for which output $\hat{h}_{m,k}=h_{\max}$.
\end{enumerate}
\item Optimize $\bm{\theta}$: Update $\bm{\theta}(t+1)=\hat{\bm{\theta}}$
based on $\mathbf{H}(t+1)$ according to (\ref{eq:cloforsolthe}).
Set $t\leftarrow t+1$ and repeat from Step \ref{enu:forasyupd} until
$\frac{1}{MK}\|\mathbf{H}(t+1)-\mathbf{H}(t)\|_{\text{F}}<\epsilon_{\text{0}}$.
\end{enumerate}
\caption{Learning the deterministic component $\bar{g}(\mathbf{p})$ \label{alg:asycoodes}}
\end{algorithm}

\subsection{Optimizing $\mathbf{H}$ via Local Polynomial Approximation\label{subsec:Optimizing--via}}

The remaining challenge is to compute $\hat{h}_{m,k}(\bm{\theta},\mathbf{H})$
in (\ref{eq:hat_h_mk}). A common approach finding the minimizer of
a quasiconvex function over a bounded interval is to perform a bisection
search for the critical point. However, as discussed after Theorem
\ref{thm:quaharbou}, the function $\bar{f}(\bm{\theta},\mathbf{H})$
appears as a staircase in each variable $h_{m,k}$, where the derivative
is zero almost everywhere, as illustrated in Fig. \ref{fig:fmkftilde}.
Moreover, the staircase function $\bar{f}(\bm{\theta},\mathbf{H})$
is not available, but only its noisy counterpart $f(\bm{\theta},\mathbf{H})$
in (\ref{eq:formulation}) is accessible to the algorithm.

We propose to smooth $f(\bm{\theta},\mathbf{H})$ without losing the
partial quasiconvex property of $\bar{f}(\bm{\theta},\mathbf{H})$.
One possibility is to employ local polynomial approximation to estimate
$\bar{f}(\bm{\theta},\mathbf{H})$ from $f(\bm{\theta},\mathbf{H})$.
Specifically, we use a polynomial to approximate $\bar{f}_{m,k}(h;\bm{\theta},\mathbf{H}_{m,k}^{-})$
in (\ref{eq:bar-f_mk}) at the neighborhood of $h_{m,k}$:
\[
\tilde{f}_{m,k}(h;\bm{a},h_{m,k})=a_{0}+a_{1}(h-h_{m,k})+a_{2}(h-h_{m,k})^{2}+\cdots
\]
where the coefficients $\bm{a}=(a_{0},a_{1},\cdots)$ are computed
by sampling $f_{m,k}(h;\bm{\theta},\mathbf{H}_{m,k}^{-})$ over a
set $\mathcal{Z}$ of scattered points $z$ in the interval $[0,H_{\max}]$
and minimizing the weighted squared error: 
\begin{equation}
\sum_{z\in\mathcal{Z}}\left(f_{m,k}(z;\bm{\theta},\mathbf{H}_{m,k}^{-})-\tilde{f}_{m,k}(z;\bm{a},h_{m,k})\right)^{2}K_{b}(z-h_{m,k}).\label{eq:locpolgra}
\end{equation}
Here, $f_{m,k}(h;\bm{\theta},\mathbf{H}_{m,k}^{-})$ is a notation
defined according to $f(\bm{\theta},\mathbf{H})$ in a way similar
to the definition of $\bar{f}_{m,k}(h;\bm{\theta},\mathbf{H}_{m,k}^{-})$
in (\ref{eq:bar-f_mk}), and is computed via (\ref{eq:formulation}).
The term $K_{b}(u)$ is a kernel function that assigns a high weight
if the distance $u=|z-h_{m,k}|$ is small, and a low weight if the
distance $u$ is large. It was found that the Epanechnikov kernel
$K_{b}(u)=\frac{3}{4b}\big(1-(\frac{u}{b})^{2}\big)_{+}$ minimizes
the asymptotic approximation error of the polynomial $\tilde{f}_{m,k}(h;\bm{a},h_{m,k})$
for a given window size $b$ \cite{Fan:b96}. Here, we can adapt $b$
according to the volume of the measurement data such that $\tilde{f}_{m,k}$
is smoothed and the gradient is non-degenerated.

It is clear that $a_{1}^{*}(h_{m,k})$, from the minimizer of (\ref{eq:locpolgra}),
is the approximated (but non-degenerated) gradient of $\bar{f}$ at
$h=h_{m,k}$. Note that it suffices to determine the sign of $a_{1}^{*}(h_{m,k})$
due to the quasiconvexity in Theorems \ref{thm:quatwosof} and \ref{thm:quaharbou}.
Specifically, we can perform a bisection search to seek the minimizer
$\hat{h}_{m,k}$ in (\ref{eq:hat_h_mk}) as follows: (i) Initialize
$h_{\min}=0$ and $h_{\max}=H_{\max}$. (ii) Set $h_{m,k}=\frac{1}{2}(h_{\min}+h_{\max})$
and find the minimizer $a_{1}^{*}$ from minimizing (\ref{eq:locpolgra}).
(iii) If $a_{1}^{*}<0$, then $h_{\min}\leftarrow h_{m,k}$; if $a_{1}^{*}>0$,
then $h_{\max}\leftarrow h_{m,k}$. (iv) Repeat from Step (ii) until
$|h_{\max}-h_{\min}|<\epsilon$, and output $\hat{h}_{m,k}=h_{\max}$.
\begin{figure}
\begin{centering}
\includegraphics[width=0.5\columnwidth]{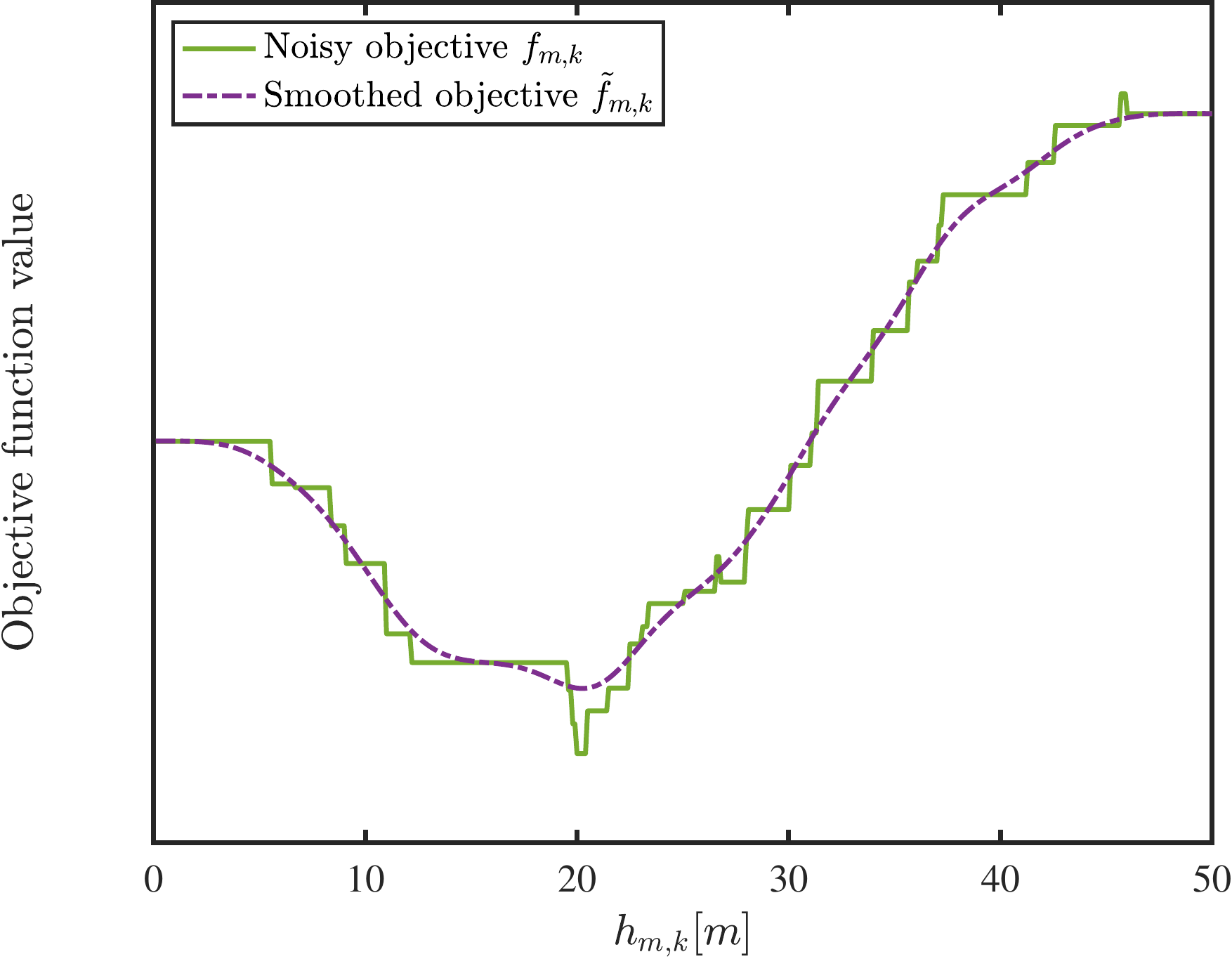}
\par\end{centering}
\caption{Illustration of $f_{m,k}$ (solid line) and the smoothed one using
local polynomial approximation (dashed line).}
\label{fig:fmkftilde}
\end{figure}

\subsection{The Overall Algorithm}

One can employ alternative optimization to repeatedly optimize $\bm{\theta}$
according to (\ref{eq:cloforsolthe}) and $\mathbf{H}$ according
to (\ref{eq:hat_h_mk}) to search for the jointly optimal solution
$(\hat{\bm{\theta}},\hat{\mathbf{H}})$ as summarized in Algorithm
\ref{alg:asycoodes}. The deterministic radio map is constructed as
\[
\hat{g}(\mathbf{p})\triangleq\bar{g}(\mathbf{p};\hat{\bm{\theta}},\hat{\mathbf{H}}).
\]

\section{Reconstructing the Shadowing using Kriging\label{sec:Reconstruction-the-Shadowing}}

After reconstructing the deterministic radio map $\hat{g}(\mathbf{p})$
from Section \ref{sec:Radio-Map-Learning}, this section focuses on
constructing the residual shadowing map $\xi(\mathbf{p})$ for the
radio map model (\ref{eq:channel-model}).\lyxdeleted{User}{Thu Sep  1 09:54:13 2022}{ }

Denote $\hat{\xi}^{(i)}=y^{(i)}-\hat{g}(\mathbf{p}^{(i)})$ as the
estimated shadowing at $\mathbf{p}^{(i)}$. The goal is to interpolate
the random process $\xi(\mathbf{p})$ based on $\hat{\xi}^{(i)}$
constructed at various locations $\mathbf{p}^{(i)}$ using the a data-driven
approach.

Recall the measurement model $y^{(i)}$ in (\ref{eq:recsigmea}) and
the noise model $n=\xi(\mathbf{p}^{(i)})+\tilde{n}$, we obtain 
\begin{align}
\hat{\xi}^{(i)} & =\bar{g}(\mathbf{p}^{(i)};\bm{\theta}^{*},\mathbf{H}^{*})+\xi(\mathbf{p}^{(i)})+\tilde{n}-\hat{g}(\mathbf{p}^{(i)})\nonumber \\
 & =\xi(\mathbf{p}^{(i)})+\tilde{n}+\big(\bar{g}(\mathbf{p}^{(i)};\bm{\theta}^{*},\mathbf{H}^{*})-\hat{g}(\mathbf{p}^{(i)})\big)\nonumber \\
 & \approx\xi(\mathbf{p}^{(i)})+\tilde{n}\label{eq:hat_xi_i}
\end{align}
where the approximation is asymptotically accurate because the term
$\bar{g}(\mathbf{p}^{(i)};\bm{\theta}^{*},\mathbf{H}^{*})-\hat{g}(\mathbf{p}^{(i)})$
tends to 0 as $N\to\infty$ according to Corollary \ref{cor:asyequrad}.

Consider constructing $\hat{\xi}(\mathbf{p})$ at $\mathbf{p}\notin\{\mathbf{p}^{(i)}\}$
as a linear combination of the measurements $\xi(\mathbf{p}^{(i)})$
\begin{equation}
\hat{\xi}(\mathbf{p})=\sum_{i=1}^{N}\lambda_{i}(\mathbf{p})\xi(\mathbf{p}^{(i)})\label{eq:xihatest}
\end{equation}
where the set of coefficients $\{\lambda_{i}(\mathbf{p})\}$ depends
on location $\mathbf{p}$. 

For a given location $\mathbf{p}$, a widely used Kriging approach
\cite{BraJemForMou:J16,SatFuj:J17} determines the coefficients $\lambda_{i}\triangleq\lambda_{i}(\mathbf{p})$
by minimizing the variance of the estimation error
\begin{equation}
\underset{\mathbf{1}^{\textrm{T}}\boldsymbol{\lambda}=1}{\mathop{\textrm{minimize}}}\quad\mathbb{V}\left\{ \hat{\xi}(\mathbf{p})-\xi(\mathbf{p})\right\} .\label{eq:xihatpro}
\end{equation}

Substituting (\ref{eq:hat_xi_i}) and (\ref{eq:xihatest}) into (\ref{eq:xihatpro}),
the objective function becomes 
\begin{align*}
\mathbb{V}\left\{ \hat{\xi}(\mathbf{p})-\xi(\mathbf{p})\right\}  & =\mathbb{V}\Big\{\sum_{i=1}^{N}\lambda_{i}(\mathbf{p})\xi(\mathbf{p}^{(i)})-\xi(\mathbf{p})\Big\}\\
 & =\mathbb{V}\Big\{\sum_{i=1}^{N}\lambda_{i}\hat{\xi}^{(i)}-\xi(\mathbf{p})\Big\}-\sum_{i=1}^{N}\lambda_{i}^{2}\mathbb{V}\left\{ \tilde{n}\right\} 
\end{align*}
where the derivation is due to the fact that the measurement noise
$\tilde{n}$ is independent of the residual shadowing process $\xi(\mathbf{p})$.

To compute the variance above, we need to build a \emph{semivariogram}
as follows.

\subsection{Semivariogram \label{subsec:Semivariogram-Function}}

Under the stationary assumption on the process $\xi(\mathbf{p})$,
the \emph{semivariogram} for $\xi(\mathbf{p})$ is defined as a function
$v(\|\mathbf{u}\|_{2})=\frac{1}{2}\mathbb{E}\{(\xi(\mathbf{p}+\mathbf{u})-\xi(\mathbf{p}))^{2}\}$.
However, the function $v(u)$ is unavailable, one needs to learn the
semivariogram model from the data. A commonly used one is the exponential
semivariogram model:
\begin{equation}
\bar{v}(u;\boldsymbol{\alpha})=\alpha_{s}^{2}\Big(1-\textrm{exp}\big(-\frac{u}{\alpha_{r}}\big)\Big)\label{eq:expsemmod}
\end{equation}
where the parameter $\boldsymbol{\alpha}=(\alpha_{s},\alpha_{r})$
can be obtained through a least-squares fitting from the data $\{\mathbf{p}^{(i)},\hat{\xi}^{(i)}\}$.
Specifically, the best model parameters can be obtained as the solution
to the following least-squares problem:
\[
\underset{\boldsymbol{\alpha}}{\mathop{\textrm{minimize}}}\quad\sum_{i,j}\Big(\bar{v}(\Vert\mathbf{p}^{(i)}-\mathbf{p}^{(j)}\Vert_{2};\boldsymbol{\alpha})-\big[\hat{\xi}^{(i)}-\hat{\xi}^{(j)}\big]^{2}\Big)^{2}.
\]

\subsection{Constructing the Residual Shadowing using Kriging}

To solve (\ref{eq:xihatpro}), we can use the Lagrange multiplier
and the Lagrange function is
\begin{multline*}
L(\boldsymbol{\lambda},\mu)=\mathbb{V}\left\{ \hat{\xi}(\mathbf{p})-\xi(\mathbf{p})\right\} +\mu(\sum_{i=1}^{N}\lambda_{i}-1)\\
=\mathbb{E}\left\{ \big(\sum_{i=1}^{N}\lambda_{i}\hat{\xi}^{(i)}-\xi(\mathbf{p})\big)^{2}\right\} -\sum_{i=1}^{N}\lambda_{i}^{2}\sigma_{\text{n}}^{2}+\mu(\sum_{i=1}^{N}\lambda_{i}-1)\\
=\sum_{i=1}^{N}\lambda_{i}\mathbb{E}\left\{ \big(\hat{\xi}^{(i)}-\xi(\mathbf{p})\big)^{2}\right\} -\frac{1}{2}\sum_{i,j}\lambda_{i}\lambda_{j}\mathbb{E}\left\{ \big(\hat{\xi}^{(i)}-\hat{\xi}^{(j)}\big)^{2}\right\} -\sum_{i=1}^{N}\lambda_{i}^{2}\sigma_{\text{n}}^{2}+\mu(\sum_{i=1}^{N}\lambda_{i}-1)
\end{multline*}
where $\mathbb{E}\Big\{\big(\hat{\xi}^{(i)}-\xi(\mathbf{p})\big)^{2}\Big\}$
and $\mathbb{E}\Big\{\big(\hat{\xi}^{(i)}-\hat{\xi}^{(j)}\big)^{2}\Big\}$
can be calculated by (\ref{eq:expsemmod}). Then, take the partial
derivatives of $L(\boldsymbol{\lambda},\mu)$ and set them to zero
\cite{SatFuj:J17}. We will get $\boldsymbol{\lambda}$ and use them
in the estimator $\hat{\xi}(\mathbf{p})=\sum_{i=1}^{N}\lambda_{i}\hat{\xi}^{(i)}$
for the shadowing component.

\section{Numerical Results\label{sec:Applications-and-Numerical}}

We study a 310 meters by 340 meters area in central Shanghai, as illustrated
in Fig. \ref{fig:virobs}. There are dozens of buildings and other
objects with heights ranging from 10 to 130 meters. Their shapes include
cubes, columns, and some irregular shapes.\footnote{The 3D city map is available at \url{https://www.openstreetmap.org}.}
We chose 100 user locations at random on the ground level, and 50,000
\ac{uav} locations uniformly at random from various altitudes. Based
on the 3D city map and the deployment of the users and \acpl{uav},
two radio map datasets are generated:\footnote{The code and dataset are available at \url{https://github.com/6wj/radiomap-uav}.}

Dataset A: The radio map is simulated according to the radio map model
$g(\mathbf{p})$ in (\ref{eq:channel-model}) with $K=1$ and path
loss parameters $(\alpha_{0},\beta_{0})=(-22,-28)$ and $(\alpha_{1},\beta_{1})=(-36,-22)$.
\Ac{iid} Gaussian measurement noise with zero mean and standard deviation
$\sigma_{\text{n}}\in\{3,7\}$ dB is added to model the shadowing
$\xi(\mathbf{p})$.

Dataset B: The radio map $g(\mathbf{p})$ is generated using Remcom
Wireless Insite, a commercial 3D ray-tracing software. Up to 6 reflections
and 1 diffraction are simulated, and other parameters are set as default.
The material of all structures is considered to be concrete. The waveforms
are chosen as narrowband sinusoidal signals at frequencies bands 2.5
GHz and 28 GHz, respectively, in different experiments.

The proposed method reconstructs a radio map with a virtual obstacle
map of roughly $M=1,200$ grid cells with 9 meter spacing between
grid points. The choice of $M$ is discussed later.

\subsection{Radio Map Reconstruction\label{subsec:Radio-Map-Reconstruction}}

We first evaluate the performance of radio map reconstruction from
Dataset A. The performance is evaluated in \ac{mae} $e=\mathbb{E}\{|\hat{g}(\mathbf{p})-g(\mathbf{p})|\}$
for the reconstructed radio map.

To make a fair comparison, the following baseline schemes are evaluated:
\begin{enumerate}
\item \Ac{knn} \cite{NiNgu:C08,DenJiaZhoCui:C18}: To construct the channel
quality at each 6D location $\mathbf{p}$, the algorithm first selects
5 measurement samples that are closest to $\mathbf{p}$ from the training
set $\{\mathbf{p}^{(i)}\}$ and form the neighbor set as $\mathcal{N}(\mathbf{p})$;
then, the channel quality at $\mathbf{p}$ is computed as $\hat{g}(\mathbf{p})=\mu^{-1}\sum_{i\in\mathcal{N}(\mathbf{p})}w(\mathbf{p},\mathbf{p}^{(i)})y^{(i)}$,
where $w(\mathbf{p},\mathbf{p}^{(i)})=\exp[-\|\mathbf{p}-\mathbf{p}^{(i)}\|_{2}^{2}/(2s^{2})]$
with a properly chosen parameter $s=55$ meters and $\mu=\sum_{i\in\mathcal{N}(\mathbf{p})}w(\mathbf{p},\mathbf{p}^{(i)})$
is a normalizing factor.
\item Kriging \cite{SatFuj:J17}: The radio map $\hat{g}(\mathbf{p})$ is
constructed based on all the measurement samples $\{(\mathbf{p}^{(i)},y^{(i)})\}$
using a similar model as in (\ref{eq:xihatest}) and the model parameters
are computed using the method in Section \ref{sec:Reconstruction-the-Shadowing}.
\end{enumerate}

Several 2D slices of reconstructed radio maps are demonstrated in
Fig. \ref{fig:radmap}. It is observed that the \ac{los}/\ac{nlos}
structure in the radio maps can be roughly reconstructed.

The left one in Fig. \ref{fig:maeall} shows the \ac{mae} for the
radio map reconstruction in terms of the number of training samples
for Dataset A, where dashed lines for $\sigma_{\text{n}}=3$ dB and
solid lines for $\sigma_{\text{n}}=7$ dB in standard deviation of
the measurement noise. It is observed that the proposed radio map
reconstruction method can reduce the \ac{mae} by 2\textendash 4 dB,
corresponding to an order of reduction in the sampling complexity,
\emph{e.g.}, the proposed method requires only 500 samples to achieve
a similar or lower \ac{mae} that is achieved by Kriging or \ac{knn}
using 5,000 samples. In addition, the proposed method is shown to
be robust to measurement noise, where a similar \ac{mae} is achieved
under either 3 dB or 7 dB noise in standard deviation.
\begin{figure}
\begin{centering}
\includegraphics[width=0.55\columnwidth]{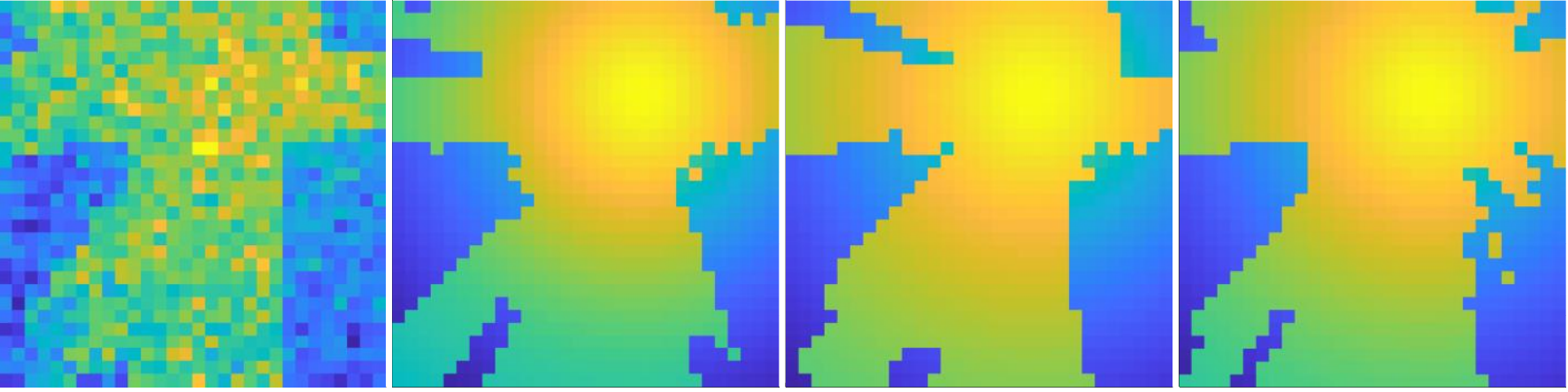}
\par\end{centering}
\caption{2D slices of radio maps from a fixed user position and fixed \ac{uav}
height, where each pixel represents the link quality between the \ac{uav}
at the corresponding $(x,y)$ position and the ground user at a fixed
position. From left to right: the true radio map, and the deterministic
radio maps $\bar{g}(\mathbf{p};\hat{\bm{\theta}},\hat{\mathbf{H}})$
reconstructed from 400, 900, 2500 measurement samples, respectively.}

\label{fig:radmap}
\end{figure}

The right one in Fig. \ref{fig:maeall} shows the \ac{mae} versus
the number of training samples from Dataset B. Two observations are
made:
\begin{figure}
\begin{centering}
\includegraphics[width=0.5\columnwidth]{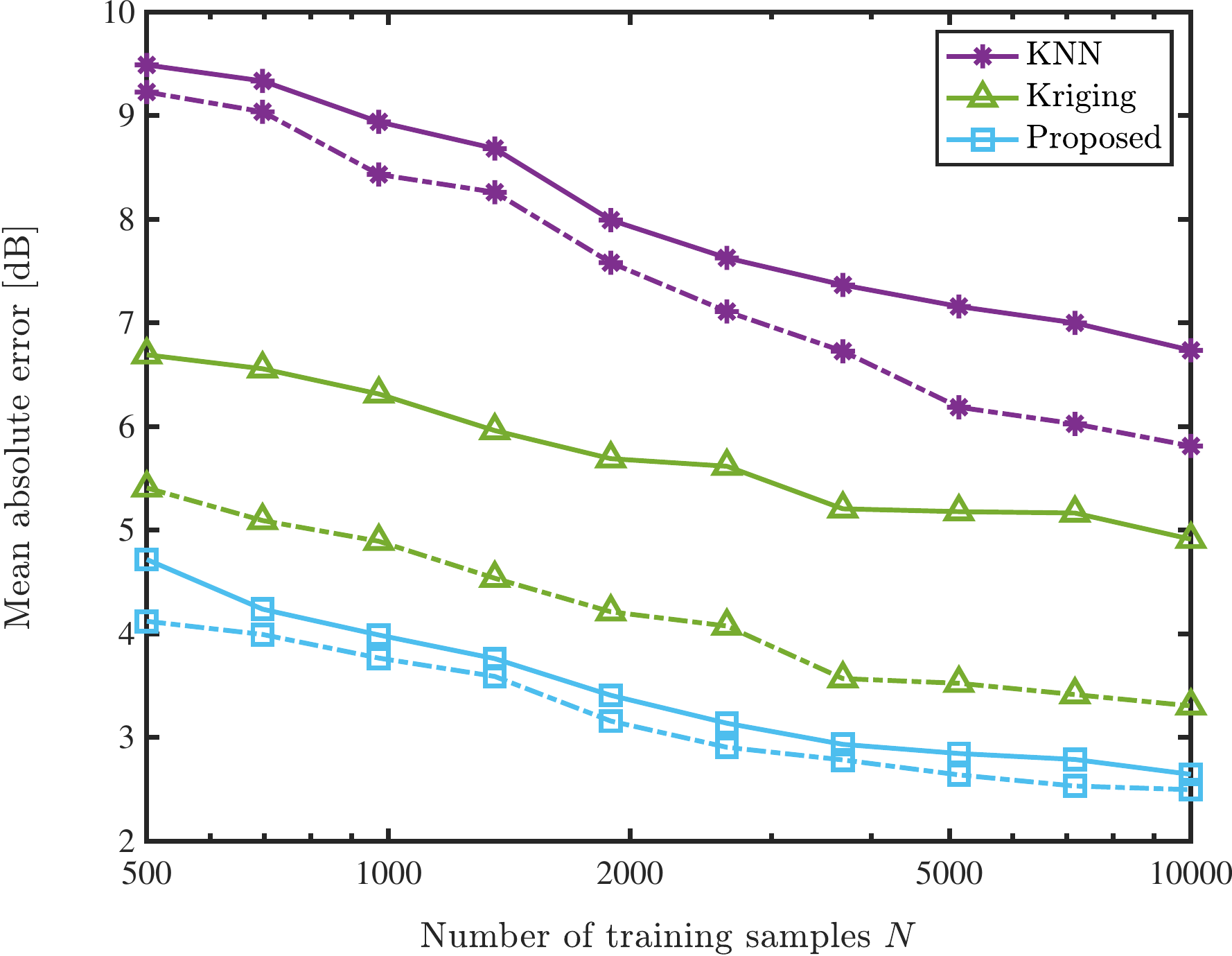}\includegraphics[width=0.5\columnwidth]{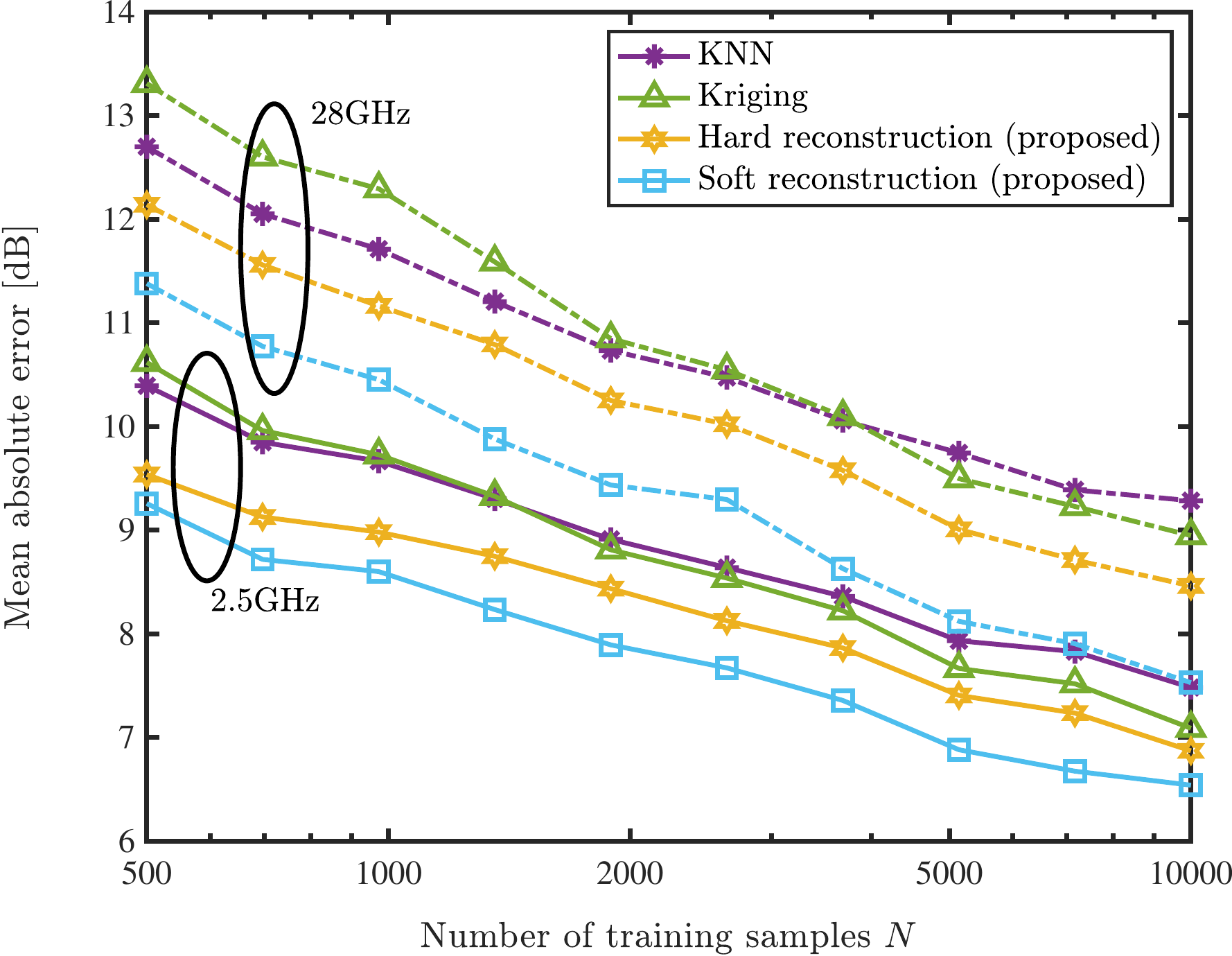}
\par\end{centering}
\caption{Left: reconstruction error versus the number of training samples $N$,
where solid lines are from 7 dB data and dashed lines are from 3 dB
data; right: reconstruction error versus the number of training samples
$N$, where solid lines are from data of the frequency at 2.5 GHz
and dashed lines are from 28 GHz data.}

\label{fig:maeall}
\end{figure}

\begin{itemize}
\item \textbf{Geometry-Awareness:} The proposed method, which estimates
both the environment-aware deterministic radio map $\bar{g}(\mathbf{p})$
and the residual shadowing $\xi(\mathbf{p})$, can bring down the
MAE by 1\textendash 2 dB, corresponding to more than 50\% reduction
in the required measurement samples. This confirms that recovering
the virtual geometry of propagation environment with radio semantics
does help radio map reconstruction.
\item \textbf{Spatial Correlation:} The proposed hard reconstruction scheme
uses the indicator function in (\ref{eq:esthidvircon}) to make a
hard decision on the propagation condition, whereas, the proposed
soft reconstruction scheme uses the likelihood function (\ref{eq:aveinddis})
for the propagation condition. Although the two methods have the same
number of parameters to estimate in the learning phase, the soft reconstruction
model (\ref{eq:aveinddis}) has a higher model complexity. Specifically,
the link status depends on a lot more virtual map parameters $h_{m,k}$
in the soft reconstruction model (\ref{eq:aveinddis}), whereas, in
the hard boundary model (\ref{eq:esthidvircon}), the status of the
same link depends on only a subset of the parameters from that of
the soft model (\ref{eq:aveinddis}); the hard model (\ref{eq:esthidvircon})
is a special case of the soft one. As a result, given enough training
data, the soft reconstruction model (\ref{eq:aveinddis}) requires
a higher computational complexity in the reconstruction phase but
achieves substantially better performance than hard reconstruction.
\end{itemize}

\subsection{Reconstructing the Geometry of the Radio Environment \label{subsec:Radio-environment-reconstruction}}

We demonstrate the recovered geometry of the surroundings that represent
the propagation environment with radio semantics.

Fig. \ref{fig:virobs} shows the environment reconstruction under
Dataset A and 7 dB measurement noise. The left figure in Fig.~\ref{fig:virobs}
shows the city map of a local area, and the right figure shows the
reconstructed virtual obstacle map from the \ac{rss} measurements.
It is observed that the geometry of the reconstructed radio environment
is roughly consistent with the city map.

Note that the two maps have different physical meanings. The city
map represents objects seen by visual light via reflection and scattering,
whereas, the virtual obstacle map represents objects ``seen'' by
radio signals via penetration, diffraction, reflection, and scattering,
etc. The virtual obstacle map servers as a low dimensional (2D) geometry
interpretation of the radio environment.

Furthermore, we obtain some insights from our experiments for the
choice of the parameters $K$ and $M$ in constructing the virtual
obstacle map.
\begin{itemize}
\item \textbf{Performance-Complexity Tradeoff}: The number of virtual obstacle
types $K$ and the number of grid cells $M$ affect the model complexity
as seen in Section \ref{subsec:Multi-class-Virtual-Obstacle}. Specifically,
the number of parameters to be estimated (such as $h_{m,k}$ in (\ref{eq:esthidvircon}))
scales as $\mathcal{O}(MK)$. In an ideal case, a larger model may
provide a better approximation to reality. For example, at $K=2$,
the proposed model not only differentiates \ac{nlos} links from \ac{los}
ones, but also differentiates whether the \ac{nlos} links suffer
from strong attenuation or light attenuation. However, a model with
a larger $K$ or $M$ also requires more training data and costs a
higher computational burden to estimate the model parameters.
\item \textbf{Spatial Resolution and Measurement Data Requirement Tradeoff}:
The parameters $K$ and $M$ also affect the spatial resolution of
the reconstructed radio map. A larger $M$ corresponds to a finer
spacing of grid cells for estimating the virtual obstacles, and a
larger $K$ corresponds to more types of virtual obstacles per grid
cell. Therefore, it is expected that a model with larger $K$ and
$M$ may represent finer details in a radio map, and likewise, require
more measurement data and more computational resources.
\end{itemize}
Our experiments suggest that the best practice for choosing $M$ and
$K$ is such that there are on average 5\textendash 20 measurement
links passing over a grid cell for each type of virtual obstacle.
Moreover, our earlier work also studied a method to dynamically adjust
the resolution $M$ locally according to the amount of measurement
data \cite{ZhaChe:C20}. 
\begin{figure}
\begin{centering}
\includegraphics[width=0.24\columnwidth]{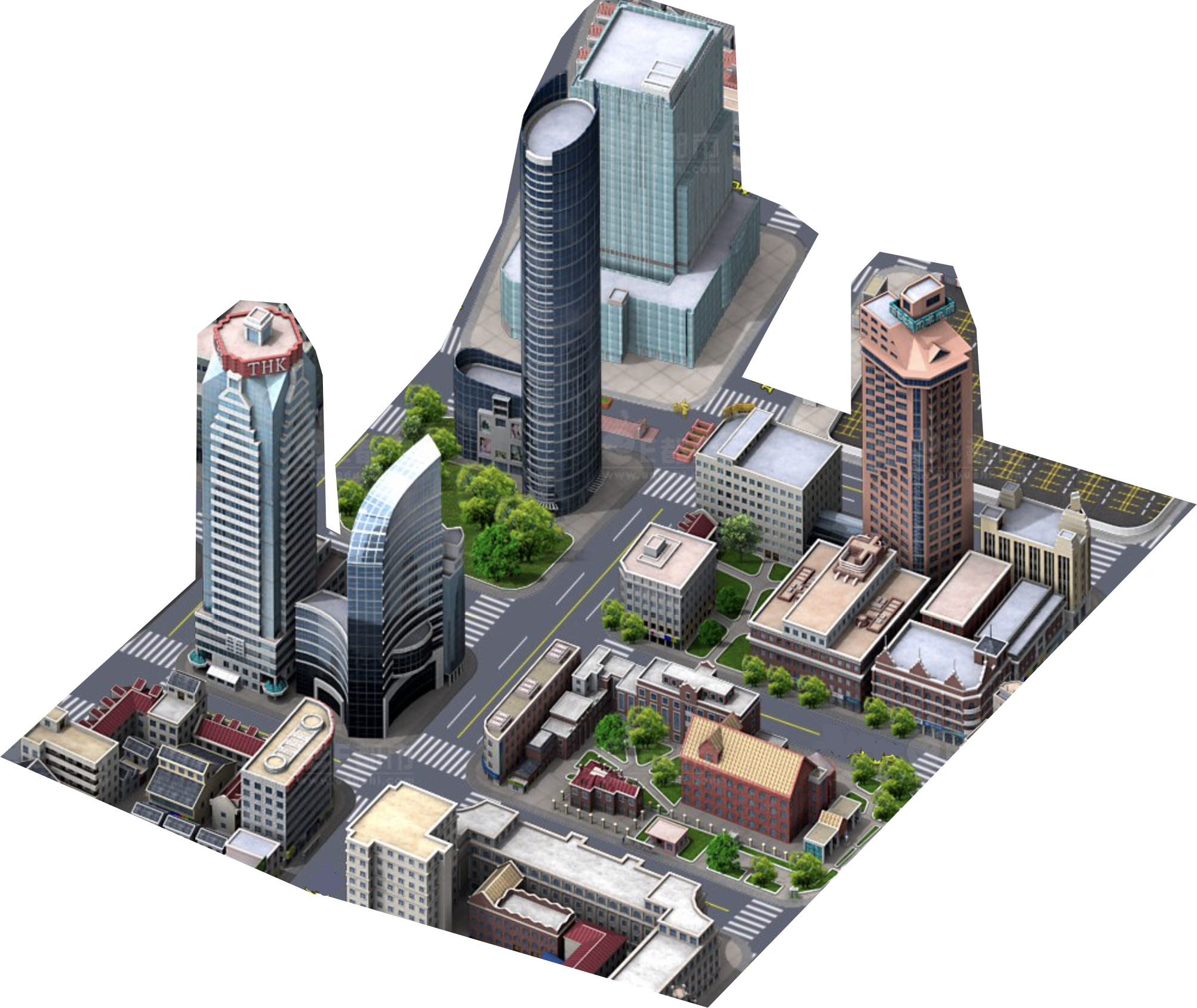}\includegraphics[width=0.25\columnwidth]{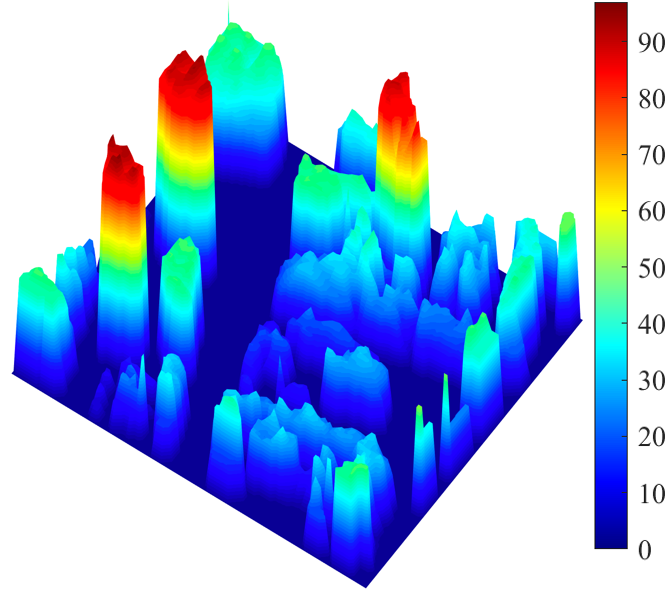}
\par\end{centering}
\caption{The left figure is the ground truth of $K=1$ virtual obstacle model.
The right figure demonstrates $K=1$ virtual obstacle height estimation
$\hat{\mathbf{H}}$. The color represents the heights in meters.}

\label{fig:virobs}
\end{figure}

\subsection{Application in UAV-aided Relay Communication}

Consider a scenario of placing a \ac{uav} relay in 3D to establish
a relay link for two ground users probably in deep shadow in a dense
urban environment. For demonstration purpose, suppose that a half-duplex
decode-and-forward relay strategy is used for a narrowband single
antenna system. We evaluate the end-to-end capacity from one user
to the other via the \ac{uav} relay. The capacity is clearly a function
of the \ac{uav} position $\mathbf{p}_{\text{d}}$ and the radio environment:
$C_{\textrm{DF}}(\mathbf{p}_{\text{d}})=\frac{1}{2}W\min\left\{ \log_{2}(1+\kappa P_{\textrm{b}}g_{\textrm{d,b}}(\mathbf{p}_{\text{d}})),\log_{2}(1+\kappa P_{\textrm{r}}g_{\textrm{u,d}}(\mathbf{p}_{\text{d}}))\right\} $,
where $W=100$ MHz is the bandwidth, $\kappa=0.5$ is a discount factor
to capture the modulation and coding loss, $P_{\textrm{b}},P_{\textrm{r}}=104$
dB is the ratio of the transmission power $20$ dBm over the received
noise power $N_{0}W$ with $N_{0}=-164$ dBm/Hz, and $g_{\textrm{d,b}}(\mathbf{p}_{\text{d}}),g_{\textrm{u,d}}(\mathbf{p}_{\text{d}})$
are channel gain depending on the \ac{uav} position. 

We propose to maximize the relay channel capacity by optimizing the
\ac{uav} position $\mathbf{p}_{\text{d}}$ using the radio map $\widetilde{g}(\mathbf{p})=\hat{g}(\mathbf{p})+\hat{\xi}(\mathbf{p})$
constructed in this paper. To benchmark the performance, we also evaluate
baselines that optimize the \ac{uav} position based on the radio
map, or propagation model, constructed by the following schemes found
in the recent literature:
\begin{enumerate}
\item Statistical Map \cite{AlhKanJam:C14,YouZha:J20}: Define $p(\phi(\mathbf{p}))$
as the \ac{los} probability of the user\textendash \ac{uav} position
pair $\mathbf{p}=(\mathbf{p}_{\text{u}},\mathbf{p}_{\text{d}})$ given
the elevation angle $\phi(\mathbf{p})$ from the user to the \ac{uav}.
Then, the channel gain $g_{\textrm{d,b}}(\mathbf{p}_{\text{d}})$
or $g_{\textrm{u,d}}(\mathbf{p}_{\text{d}})$ is estimated as $p(\phi(\mathbf{p}))G_{0}(\mathbf{p})+(1-p(\phi(\mathbf{p})))G_{1}(\mathbf{p})$,
where $G_{k}(\mathbf{p})=a_{k}+b_{k}\log_{10}\tilde{d}(\mathbf{p})$
with $\tilde{d}(\mathbf{p})$ being the propagation distance. The
empirical distribution $p(\phi)$ is obtained via offline training
based on data $\{(\mathbf{p}^{(i)},y^{(i)})\}_{i=1}^{N}$, and the
parameters $a_{k},b_{k}$ are empirically fitted from data under \ac{los}
(for $k=0$) or under \ac{nlos} (for $k=1$).\lyxdeleted{User}{Thu Sep  1 09:54:13 2022}{ }
\item \Ac{knn}-based Radio Map: The radio map is constructed using the
\ac{knn} approach described in Section \ref{subsec:Radio-Map-Reconstruction}.\lyxdeleted{User}{Thu Sep  1 09:54:13 2022}{ }
\end{enumerate}
\begin{figure}
\begin{centering}
\includegraphics[width=0.5\columnwidth]{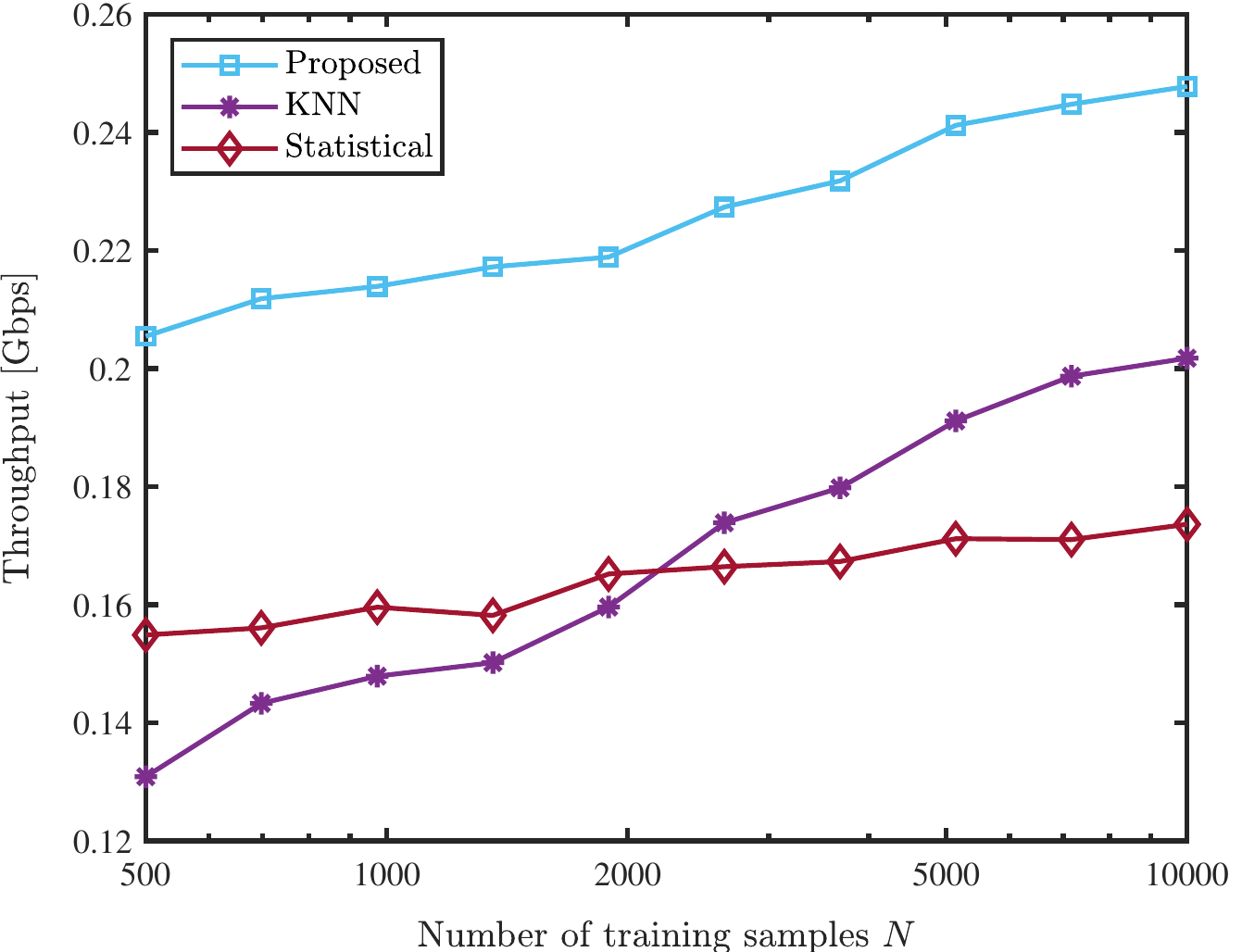}
\par\end{centering}
\caption{Average \ac{uav}-aided relay channel capacity evaluated over 1,225
user positions from Dataset B in a dense urban environment.\lyxdeleted{User}{Thu Sep  1 09:54:13 2022}{ }}

\label{fig:maxthr}
\end{figure}

Fig. \ref{fig:maxthr} demonstrates the average capacity versus the
number of training samples $N$ from Dataset B for the 2.5 GHz case.
Statistical method requires the least number of training samples before
its performance saturates, because the model has the least parameters
and a poor capability to describe the radio environment. The capacity
of both the \ac{knn}-based approach and the proposed approach increases
as the number of training samples $N$ gets larger. The significant
gain of the proposed scheme over the \ac{knn}-based approach can
be understood from the fact that the proposed method reconstructs
a radio map much more accurate than the \ac{knn} approach does. Specifically,
the proposed scheme achieves more than $50\%$ capacity gain over
the \ac{knn}-based approach in the small sample regime (around 500
training samples); in addition, it also outperforms the statistical
scheme by more than $30\%$ capacity gain and the performance gain
further increases when more training samples are used. Thus, we may
conclude that an accurate radio map may substantially enhance the
performance in a wireless communication system.

\section{Conclusion\label{sec:Conclusion}}

This paper developed a radio map model and estimation algorithms for
constructing radio maps from \ac{rss} measurements between aerial
nodes and ground nodes. The core idea is to construct a multi-class
3D virtual obstacle model from the \ac{rss} measurement data to capture
the geometry of the environment with radio semantics. A joint estimation
problem of the propagation parameters and the 3D virtual environment
map is formulated. While the estimation problem is non-convex with
degenerated gradient, we discover the \lyxdeleted{User}{Thu Sep  1 09:54:13 2022}{ }partially
quasiconvex property, which leads to the development of an efficient
parameter estimation radio map construction algorithm. Numerical results
demonstrated that by constructing the 3D virtual environment, the
required measurement data for achieving a same radio map construction
accuracy can be reduced by more than a half. It is also found that
when the proposed radio map is applied to \ac{uav} placement for
\ac{uav}-aided relay communication, more than $50\%$ relay capacity
gain can be achieved.
\begin{appendices}

\section{Proof of Lemma 1}\label{app:prooflemnoi}

Under $y^{(i)}=\bar{g}(\mathbf{p}^{(i)};\bm{\theta}^{*},\mathbf{H}^{*})+n^{(i)}$,
we obtain
\begin{align}
f(\bm{\theta},\mathbf{H}) & =\frac{1}{N}\sum\limits _{i=1}^{N}\Big[y^{(i)}-\sum_{k=0}^{K}\big(\beta_{k}+\alpha_{k}d(\mathbf{p}^{(i)})\big)S_{k}(\mathbf{p}^{(i)};\mathbf{H})\Big]^{2}\nonumber \\
 & =\frac{1}{N}\sum\limits _{i=1}^{N}\Big[\bar{g}(\mathbf{p}^{(i)};\bm{\theta}^{*},\mathbf{H}^{*})+n^{(i)}-\sum_{k=0}^{K}\big(\beta_{k}+\alpha_{k}d(\mathbf{p}^{(i)})\big)S_{k}(\mathbf{p}^{(i)};\mathbf{H})\Big]^{2}\nonumber \\
 & =\frac{1}{N}\sum\limits _{i=1}^{N}\left[\big(\Delta^{(i)}\big)^{2}+2\Delta^{(i)}n^{(i)}+\big(n^{(i)}\big)^{2}\right]\label{eq:app-lemma-eq1}\\
 & =\bar{f}(\bm{\theta},\mathbf{H})+\frac{1}{N}\sum\limits _{i=1}^{N}\big[2\Delta^{(i)}n^{(i)}+\big(n^{(i)}\big)^{2}\big]\label{eq:app-lemma-eq2}
\end{align}
where, in (\ref{eq:app-lemma-eq1}), we denote 
\[
\Delta^{(i)}=\bar{g}(\mathbf{p}^{(i)};\bm{\theta}^{*},\mathbf{H}^{*})-\sum_{k=0}^{K}\big(\beta_{k}+\alpha_{k}d(\mathbf{p}^{(i)})\big)S_{k}(\mathbf{p}^{(i)};\mathbf{H})
\]
and thus, $\bar{f}(\bm{\theta},\mathbf{H})=\frac{1}{N}\sum_{i=1}^{N}\Delta^{(i)}$
in (\ref{eq:app-lemma-eq2}) as from (\ref{eq:formulnonoi}).

It remains to show that the second term in (\ref{eq:app-lemma-eq2})
converges to $C$ in probability.

We first show the convergence of $\frac{1}{N}\sum_{i=1}^{N}2\Delta^{(i)}n^{(i)}$.
For brevity, we denote $X_{i}=2\Delta^{(i)}n^{(i)}=2\Delta^{(i)}(\xi^{(i)}+\tilde{n}^{(i)})$
from the observation model, where we have further denoted $\xi^{(i)}\triangleq\xi(\mathbf{p}^{(i)})$.
It follows that $\mathbb{E}\{X_{i}\}=0$ since both $\xi^{(i)}$ and
$\tilde{n}^{(i)}$ have zero mean. As a result, $\mbox{cov}(X_{i},X_{j})=\mathbb{E}\{X_{i}X_{j}\}=4\Delta^{(i)}\Delta^{(j)}\mbox{cov}(\xi^{(i)},\xi^{(j)})$
for $i\neq j$, since $\xi^{(i)}$ and $\tilde{n}^{(i)}$ are independent
with zero mean.

Note that $\Delta^{(i)}$ is a deterministic function with bounded
values due to its definition and the fact that the regions of $\bm{\theta}$,
$\mathbf{H}$, and $\mathbf{p}^{(i)}$ are bounded, \emph{i.e.}, there
exists $C_{0}$, such that $4|\Delta^{(i)}\Delta^{(j)}|<C_{0}$ for
all $i,j$. Therefore, we must have $|\mbox{cov}(X_{i},X_{j})|\leq C_{0}|\mbox{cov}(\xi^{(i)},\xi^{(j)})|$
for $i\neq j$. Since $|\mbox{cov}(\xi^{(i)},\xi^{(j)})|\to0$ as
$|i-j|\to\infty$ and $N\to\infty$, it follows that $|\mbox{cov}(X_{i},X_{j})|\to0$.

We are interested in the variance $\mathbb{V}\{\frac{1}{N}\sum_{i=1}^{N}X_{i}\}$,
which equals to 
\begin{equation}
\frac{1}{N^{2}}\sum_{i=1}^{N}\mathbb{V}\{X_{i}\}+\frac{2}{N^{2}}\sum_{i=1}^{N}\sum_{j=i+1}^{N}\mbox{cov}(X_{i},X_{j})\label{eq:app-lemma-var}
\end{equation}
where, in the first term, $\mathbb{V}\{X_{i}\}\leq C_{1}$ since $\xi^{(i)}$
and $\tilde{n}^{(i)}$ have bounded variance and $\Delta^{(i)}$ is
bounded. In the second term, the condition $|\mbox{cov}(X_{i},X_{j})|\to0$
implies that for any $\epsilon_{1}>0$, there exists a finite $N_{1}$
such that $\mbox{cov}(X_{i},X_{j})<\epsilon_{1}$ for all $j-i>N_{1}$.
So, it follows that 
\begin{align*}
\sum_{j=i+1}^{N}\mbox{cov}(X_{i},X_{j}) & =\sum_{j=i+1}^{i+N_{1}}\mbox{cov}(X_{i},X_{j})+\sum_{j=i+N_{1}+1}^{N}\mbox{cov}(X_{i},X_{j})\\
 & \leq\sum_{j=i+1}^{i+N_{1}}\sqrt{\mathbb{V}\{X_{i}\}\mathbb{V}\{X_{j}\}}+\sum_{j=i+N_{1}+1}^{N}\epsilon_{1}\\
 & \leq N_{1}C_{1}+N\epsilon_{1}
\end{align*}
where the first term in the second line is due to the Cauchy-Schwarz
inequality.

As a result, $\mathbb{V}\{\frac{1}{N}\sum_{i=1}^{N}X_{i}\}$ from
(\ref{eq:app-lemma-var}) can be upper bounded as 
\begin{align*}
\mathbb{V}\Big\{\frac{1}{N}\sum_{i=1}^{N}X_{i}\Big\} & \leq\frac{1}{N^{2}}\sum_{i=1}^{N}C_{1}+\frac{2}{N^{2}}\sum_{i=1}^{N}\big(N_{1}C_{1}+N\epsilon_{1}\big)\\
 & =\frac{(1+2N_{1})C_{1}}{N}+2\epsilon_{1}.
\end{align*}

Therefore, for any $\epsilon>0$, one can choose $\epsilon_{1}=\epsilon/3$
and $N_{2}=3(1+2N_{1})C_{1}/\epsilon$, such that $\mathbb{V}\{\frac{1}{N}\sum_{i=1}^{N}X_{i}\}<\epsilon$
for all $N>N_{2}$. This establishes that $\mathbb{V}\{\frac{1}{N}\sum_{i=1}^{N}X_{i}\}\to0$
as $N\to\infty$.

As a result, by Chebyshev's inequality, we have 
\[
\mathbb{P}\Big\{\Big|\frac{1}{N}\sum_{i=1}^{N}X_{i}-\mathbb{E}\{X_{1}\}\Big|>\epsilon_{0}\}\leq\frac{\mathbb{V}\big\{\frac{1}{N}\sum_{i=1}^{N}X_{i}\big\}}{\epsilon_{0}^{2}}\to0
\]
as $N\to\infty$. This shows that $\frac{1}{N}\sum_{i=1}^{N}X_{i}\to0$
in probability.

Finally, to show the convergence of $\frac{1}{N}\sum_{i=1}^{N}\big(n^{(i)}\big)^{2}$
in (\ref{eq:app-lemma-eq2}), we have 
\[
\frac{1}{N}\sum\limits _{i=1}^{N}\big(n^{(i)}\big)^{2}=\frac{1}{N}\sum\limits _{i=1}^{N}(\tilde{n}^{(i)})^{2}+2\frac{1}{N}\sum\limits _{i=1}^{N}\tilde{n}^{(i)}\xi^{(i)}+\frac{1}{N}\sum\limits _{i=1}^{N}(\xi^{(i)})^{2}
\]
where the first term and the second term respectively converges to
$\mathbb{E}\{(\tilde{n}^{(i)})^{2}\}=\sigma_{\text{n}}^{2}$ and $\mathbb{E}\{\tilde{n}^{(i)}\xi^{(i)}\}=0$
in probability, as $N\to\infty$, due to the weak law of large number
and the fact that $\tilde{n}^{(i)}$ and $\xi^{(i)}$ are independent;\footnote{More rigorously, the convergence of the second term can be proven
by Chebyshev's inequality following a similar procedure as proving
the convergence of $\frac{1}{N}\sum_{i=1}^{N}X_{i}$.} the third term converges to a finite value $C_{2}$ from the assumption
of the lemma.

Therefore, we have shown that $f(\bm{\theta},\mathbf{H})\to\bar{f}(\bm{\theta},\mathbf{H})+C$
for $C=\sigma_{\text{n}}^{2}+C_{2}$ in probability as $N\to\infty$
for every $(\bm{\theta},\mathbf{H})$.

\section{Proof of Theorem 1}\label{app:proofthebas}

To simplify the notations of the proof, we denote $\gamma_{k}^{(i)}=\beta_{k}^{*}+\alpha_{k}^{*}d(\mathbf{p}^{(i)})$
as the path loss for $d(\mathbf{p}^{(i)})$ and the $k$th degree
of signal obstruction, $\mathbb{I}_{k}^{(i,j)}(\mathbf{h})=\mathbb{I}\{\mathbf{p}^{(i)}+\bm{\epsilon}_{j}\in\mathcal{D}_{k}(\mathbf{h})\},\forall j=1,\ldots,J-1$,
$\mathbb{I}_{k}^{(i)}(\mathbf{h})=\mathbb{I}\{\mathbf{p}^{(i)}\in\mathcal{D}_{k}(\mathbf{h})\}$.
Accordingly, $\bar{g}^{(i)}\triangleq\bar{g}(\mathbf{p}^{(i)};\bm{\theta}^{*},\mathbf{h}^{*})=\sum\nolimits _{k=0}^{1}\gamma_{k}^{(i)}S_{k}(\mathbf{p}^{(i)};\mathbf{h}^{*})$
from (\ref{eq:channel-model-nonoi}), and $S_{k}(\mathbf{p}^{(i)};\mathbf{h}^{*})=\sum_{j=0}^{J-1}\omega_{j}\mathbb{I}_{k}^{(i,j)}(\mathbf{h}^{*})$
from (\ref{eq:aveinddis}).

Define variable $\mathbf{h}_{-m}$ as a vector from $\mathbf{h}$
but with the $m$th element $h_{m}$ removed. Thus, by restricting
$\bar{f}(\bm{\theta},\mathbf{h})$ in (\ref{eq:formulnonoi}) to take
value in the interval $\mathcal{I}_{m}$, we have
\begin{align}
\bar{f}_{m}(h_{m};\bm{\theta}^{*},\mathbf{h}_{-m}) & =\frac{1}{N}\sum\limits _{i=1}^{N}\big[\bar{g}^{(i)}-\sum\limits _{k=0}^{1}\gamma_{k}^{(i)}\sum_{j=0}^{J-1}\omega_{j}\mathbb{I}_{k}^{(i,j)}(\mathbf{h})\big]^{2}\nonumber \\
 & =\frac{1}{N}\sum\limits _{i=1}^{N}\big[\bar{g}^{(i)}-\sum\limits _{k=0}^{1}\gamma_{k}^{(i)}\omega_{0}\mathbb{I}_{k}^{(i)}(\mathbf{h})-\sum\limits _{k=0}^{1}\gamma_{k}^{(i)}\sum_{j=1}^{J-1}\omega_{j}\mathbb{I}_{k}^{(i,j)}(\mathbf{h})\big]^{2}.\label{eq:formulsim}
\end{align}
Substituting $\bar{g}^{(i)}$ into (\ref{eq:formulsim}), we obtain
\begin{multline}
\bar{f}_{m}(h_{m};\bm{\theta}^{*},\mathbf{h}_{-m})=\frac{1}{N}\sum\limits _{i=1}^{N}\bigg[\sum\limits _{k=0}^{1}\gamma_{k}^{(i)}\omega_{0}\mathbb{I}_{k}^{(i)}(\mathbf{h}^{*})-\sum\limits _{k=0}^{1}\gamma_{k}^{(i)}\omega_{0}\mathbb{I}_{k}^{(i)}(\mathbf{h})\\
\qquad+\sum\limits _{k=0}^{1}\gamma_{k}^{(i)}\sum_{j=1}^{J-1}\omega_{j}\mathbb{I}_{k}^{(i,j)}(\mathbf{h}^{*})-\sum\limits _{k=0}^{1}\gamma_{k}^{(i)}\sum_{j=1}^{J-1}\omega_{j}\mathbb{I}_{k}^{(i,j)}(\mathbf{h})\bigg]^{2}\\
=\frac{1}{N}\sum\limits _{i=1}^{N}\bigg[\underbrace{\omega_{0}\sum\limits _{k=0}^{1}\gamma_{k}^{(i)}(\mathbb{I}_{k}^{(i)}(\mathbf{h}^{*})-\mathbb{I}_{k}^{(i)}(\mathbf{h}))}_{(a)}+\underbrace{\sum_{j=1}^{J-1}\omega_{j}\sum\limits _{k=0}^{1}\gamma_{k}^{(i)}(\mathbb{I}_{k}^{(i,j)}(\mathbf{h}^{*})-\mathbb{I}_{k}^{(i,j)}(\mathbf{h}))}_{(b)}\bigg]^{2}.\label{eq:formulsimab}
\end{multline}

We have three cases for the value of $\sum\nolimits _{k=0}^{1}\gamma_{k}^{(i)}(\mathbb{I}_{k}^{(i)}(\mathbf{h}^{*})-\mathbb{I}_{k}^{(i)}(\mathbf{h}))$
or $\sum\nolimits _{k=0}^{1}\gamma_{k}^{(i)}(\mathbb{I}_{k}^{(i,j)}(\mathbf{h}^{*})-\mathbb{I}_{k}^{(i,j)}(\mathbf{h}))$,
\emph{i.e.}, $0$, $\gamma_{0}^{(i)}-\gamma_{1}^{(i)}$, and $\gamma_{1}^{(i)}-\gamma_{0}^{(i)}$.
Notice that, $\gamma_{0}^{(i)}>\gamma_{1}^{(i)}$ from our model (\ref{eq:channel-model})
and smaller $k$ means less signal obstruction. $\mathbb{I}_{k}^{(i,j)}(\mathbf{h}^{*})-\mathbb{I}_{k}^{(i,j)}(\mathbf{h}))\in\{-1,0,1\}$.
\begin{itemize}
\item For term (a): We have, for both $k=0$ and $1$, term (a) $\in\{0,\omega_{0}(\gamma_{0}^{(i)}-\gamma_{1}^{(i)}),\omega_{0}(\gamma_{1}^{(i)}-\gamma_{0}^{(i)})\}$.
\item For term (b): We obtain a lower bound and an upper bound as
\begin{equation}
\sum_{j=1}^{J-1}\omega_{j}(\gamma_{1}^{(i)}-\gamma_{0}^{(i)})\le(b)\le\sum_{j=1}^{J-1}\omega_{j}(\gamma_{0}^{(i)}-\gamma_{1}^{(i)}).\label{eq:deltays}
\end{equation}
\end{itemize}
Set $\omega_{0}\ge\frac{2}{3}$, it follows that $0\le\sum_{j=1}^{J-1}\omega_{j}\le\frac{1}{3}$
since $\sum_{j}\omega_{j}=1$, and hence,
\begin{equation}
[\pm\omega_{0}+\sum_{j=1}^{J-1}\omega_{j}]^{2}\ge[\sum_{j=1}^{J-1}\omega_{j}]^{2}.\label{eq:winequraw}
\end{equation}
By (\ref{eq:deltays}) and (\ref{eq:winequraw}), for any $i$, \ac{wrt}
the value of term (a),
\begin{itemize}
\item Substituting the upper bound in (\ref{eq:deltays}) into $[(a)+(b)]^{2}$
in (\ref{eq:formulsimab}), we obtain
\[
[\pm\omega_{0}(\gamma_{0}^{(i)}-\gamma_{1}^{(i)})+\sum_{j=1}^{J-1}\omega_{j}(\gamma_{0}^{(i)}-\gamma_{1}^{(i)})]^{2}\ge[0+\sum_{j=1}^{J-1}\omega_{j}(\gamma_{0}^{(i)}-\gamma_{1}^{(i)})]^{2}.
\]
\item Substituting the lower bound in (\ref{eq:deltays}) into $[(a)+(b)]^{2}$
in (\ref{eq:formulsimab}), we obtain
\[
[\pm\omega_{0}(\gamma_{0}^{(i)}-\gamma_{1}^{(i)})+\sum_{j=1}^{J-1}\omega_{j}(\gamma_{1}^{(i)}-\gamma_{0}^{(i)})]^{2}\ge[0+\sum_{j=1}^{J-1}\omega_{j}(\gamma_{1}^{(i)}-\gamma_{0}^{(i)})]^{2}.
\]
\end{itemize}
Hence,
\begin{multline}
[\pm\omega_{0}(\gamma_{0}^{(i)}-\gamma_{1}^{(i)})+\sum_{j=1}^{J-1}\omega_{j}\sum\limits _{k=0}^{1}\gamma_{k}^{(i)}(\mathbb{I}_{k}^{(i,j)}(\mathbf{h}^{*})-\mathbb{I}_{k}^{(i,j)}(\mathbf{h}))]^{2}\\
\ge[0+\sum_{j=1}^{J-1}\omega_{j}\sum\limits _{k=0}^{1}\gamma_{k}^{(i)}(\mathbb{I}_{k}^{(i,j)}(\mathbf{h}^{*})-\mathbb{I}_{k}^{(i,j)}(\mathbf{h}))]^{2}\label{eq:winequab}
\end{multline}
where $\pm\omega_{0}(\gamma_{0}^{(i)}-\gamma_{1}^{(i)})$ is for $\mathbb{I}_{k}^{(i)}(\mathbf{h}^{*})\neq\mathbb{I}_{k}^{(i)}(\mathbf{h})$
and $0$ is for $\mathbb{I}_{k}^{(i)}(\mathbf{h}^{*})=\mathbb{I}_{k}^{(i)}(\mathbf{h})$
in term (a).

Denote the right hand side in (\ref{eq:winequab}) as $A^{(i)}$,
the left hand side as $B^{(i)}$, and two sets $\mathcal{I}_{A}(\mathbf{h})=\{i\mid\mathbb{I}_{k}^{(i)}(\mathbf{h})=\mathbb{I}_{k}^{(i)}(\mathbf{h}^{*})\}$,
$\mathcal{I}_{B}(\mathbf{h})=\{i\mid\mathbb{I}_{k}^{(i)}(\mathbf{h})\neq\mathbb{I}_{k}^{(i)}(\mathbf{h}^{*})\}=\{1,2,\ldots,N\}\setminus\mathcal{I}_{A}(\mathbf{h})$.

Recall the obstacle indicator function in (\ref{eq:esthidvircon})
\[
\mathbb{I}\{\mathbf{p}^{(i)}\in\mathcal{D}_{k}(\mathbf{H})\}=(1-\underset{j\in\mathscr{\mathscr{\mathcal{B}}}^{(i)}}{\prod}(1-\mathbb{I}\{h_{j,k}\ge z_{j}^{(i)}\}))\underset{j\in\mathscr{\mathscr{\mathcal{B}}}^{(i)}}{\prod}\underset{l>k}{\prod}\mathbb{I}\{h_{j,l}<z_{j}^{(i)}\}
\]
and we denote a set
\[
\mathcal{L}_{m}(\mathbf{h})=\left\{ i\mid h_{j}<z_{j}^{(i)},\forall j\in\mathcal{B}^{(i)}\setminus\{m\}\right\} 
\]
as measurement samples that can only be blocked by $m$th grid. Then,
\ac{wrt} $h_{m}$, values of $\mathbb{I}_{k}^{(i)}(\mathbf{h})$
have two types: (i) dependent of $h_{m}$, \emph{i.e}., index $i\in\mathcal{L}_{m}(\mathbf{h})$,
and by the condition $h_{j}^{*}<h_{j}$ and $h_{j}<z_{j}^{(i)},\forall j\in\mathcal{B}^{(i)}\setminus\{m\}$,
\[
\mathbb{I}_{k}^{(i)}(h_{m};\mathbf{h}_{-m})=(1-(1-\mathbb{I}\{h_{m,k}\ge z_{m,k}^{(i)}\}))\underset{l>k}{\prod}\mathbb{I}\{h_{m,l}<z_{m}^{(i)}\}
\]
so $\mathbb{I}_{0}^{(i)}(h_{m};\mathbf{h}_{-m})=\mathbb{I}\{h_{m}<z_{m}^{(i)}\}=\mathbb{I}_{0}^{(i)}(h_{m};\mathbf{h}_{-m}^{*})$
and $\mathbb{I}_{1}^{(i)}(h_{m};\mathbf{h}_{-m})=\mathbb{I}\{h_{m}\ge z_{m}^{(i)}\}=\mathbb{I}_{1}^{(i)}(h_{m};\mathbf{h}_{-m}^{*})$;
(ii) independent of $h_{m}$, \emph{i.e}., index $i\in\{1,2,\ldots,N\}\setminus\mathcal{L}_{m}(\mathbf{h})$,
and then $\mathbb{I}_{k}^{(i)}(h_{m};\mathbf{h}_{-m})$ will be constant.
Hence, we only need to consider the type (i) samples. 
\begin{figure}
\begin{centering}
\includegraphics[width=0.5\columnwidth]{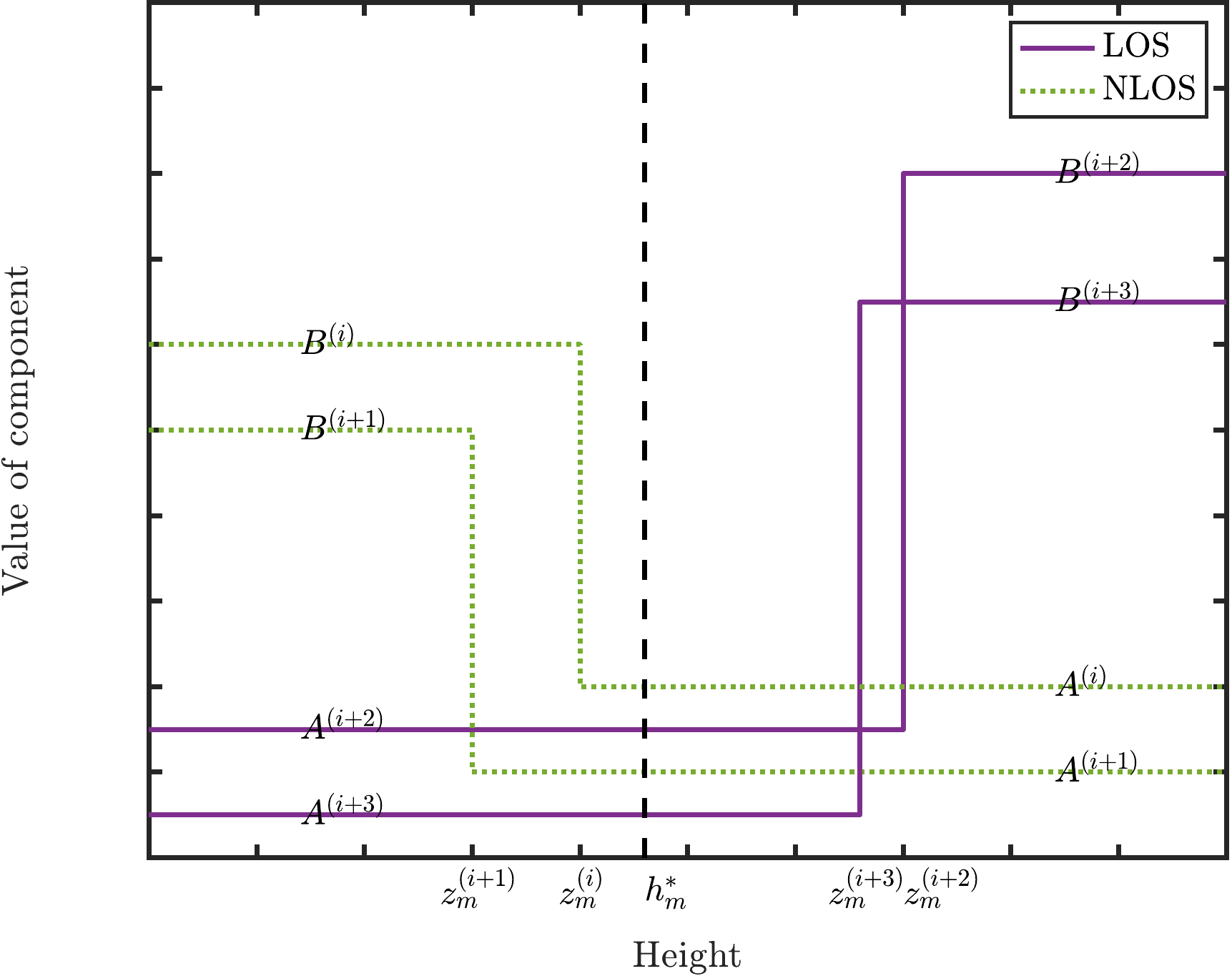}
\par\end{centering}
\caption{An example for the value of component $[\bar{g}^{(i)}-\sum\nolimits _{k=0}^{1}\gamma_{k}^{(i)}\sum\nolimits _{j=0}^{J-1}\omega_{j}\mathbb{I}_{k}^{(i,j)}(\mathbf{h})]^{2}$
vs. $h_{m}$.}

\label{fig:indobj}
\end{figure}
For index $i\in\mathcal{L}_{m}(\mathbf{h})$,
\begin{itemize}
\item \ac{los} measurement samples satisfy $h_{m}^{*}<z_{m}^{(i)}$. As
shown in Fig. \ref{fig:indobj}, when $h_{m}<z_{m}^{(i)}$, we obtain
$\mathbb{I}_{k}^{(i)}(\mathbf{h})=\mathbb{I}_{k}^{(i)}(\mathbf{h}^{*})$
and $[(a)+(b)]^{2}$ in (\ref{eq:formulsimab}) taking value $A^{(i)}$.
When $h_{m}\ge z_{m}^{(i)}$, we obtain $\mathbb{I}_{k}^{(i)}(\mathbf{h})\neq\mathbb{I}_{k}^{(i)}(\mathbf{h}^{*})$
and $[(a)+(b)]^{2}$ taking value $B^{(i)}$.
\item \ac{nlos} measurement samples satisfy $h_{m}^{*}\ge z_{m}^{(i)}$.
As shown in Fig. \ref{fig:indobj}, when $h_{m}\ge z_{m}^{(i)}$,
we obtain $\mathbb{I}_{k}^{(i)}(\mathbf{h})=\mathbb{I}_{k}^{(i)}(\mathbf{h}^{*})$
and $[(a)+(b)]^{2}$ in (\ref{eq:formulsimab}) taking value $A^{(i)}$.
When $h_{m}<z_{m}^{(i)}$, we obtain $\mathbb{I}_{k}^{(i)}(\mathbf{h})\neq\mathbb{I}_{k}^{(i)}(\mathbf{h}^{*})$
and $[(a)+(b)]^{2}$ taking value $B^{(i)}$.
\end{itemize}
Therefore, the smaller the distance between $h_{m}$ and $h_{m}^{*}$
is, the larger the set $\mathcal{I}_{A}(h_{m};\mathbf{h}_{-m})$ will
be, \emph{i.e}., the more $\mathbb{I}_{k}^{(i)}(h_{m};\mathbf{h}_{-m})=\mathbb{I}_{k}^{(i)}(h_{m}^{*};\mathbf{h}_{-m}^{*})$.
$\ensuremath{\forall}\ensuremath{\epsilon}\ensuremath{>0}$, when
$h_{m}<h_{m}^{*}$, $\mathcal{I}_{A}(h_{m}-\epsilon;\mathbf{h}_{-m})\subseteq\mathcal{I}_{A}(h_{m};\mathbf{h}_{-m})$,
and when $h_{m}>h_{m}^{*}$, $\mathcal{I}_{A}(h_{m}+\epsilon;\mathbf{h}_{-m})\subseteq\mathcal{I}_{A}(h_{m};\mathbf{h}_{-m})$.
Intuitively, for index $i\in\mathcal{L}_{m}(\mathbf{h})$, the value
of $[(a)+(b)]^{2}$ in (\ref{eq:formulsimab}) is shown in Fig. \ref{fig:indobj},
and for index $i\in\{1,2,\ldots,N\}\setminus\mathcal{L}_{m}(\mathbf{h})$,
the value is constant in $h_{m}$. We can see the summation of all
\ac{los} is increasing in $h_{m}$ and is constant from zero to a
height which is greater than $h_{m}^{*}$. Similarly, the summation
of all \ac{nlos} is decreasing in $h_{m}$ and is constant from zero
to a height which is less than $h_{m}^{*}$, so the summation of all
\ac{los} and \ac{nlos} is quasiconvex. Then, we will show the relation
between the sets and the function value mathematically.
\[
\bar{f}_{m}(h_{m};\bm{\theta}^{*},\mathbf{h}_{-m})=\underset{i\in\mathcal{I}_{A}(\mathbf{h})}{\sum}A^{(i)}+\underset{i\in\mathcal{I}_{B}(\mathbf{h})}{\sum}B^{(i)}
\]
where $0<A^{(i)}\le B^{(i)}$. $\text{\ensuremath{\forall}\ensuremath{\epsilon}\ensuremath{>0}}$,
when $\forall h_{m}<h_{m}^{*}$, by using $\mathcal{I}_{A}(h_{m}-\epsilon;\mathbf{h}_{-m})\subseteq\mathcal{I}_{A}(h_{m};\mathbf{h}_{-m})$
and $\mathcal{I}_{A}(h_{m};\mathbf{h}_{-m})=\{1,2,\ldots,N\}\setminus\mathcal{I}_{A}(h_{m};\mathbf{h}_{-m})$,
we finally have
\[
\bar{f}_{m}(h_{m}-\ensuremath{\epsilon};\bm{\theta}^{*},\mathbf{h}_{-m})\ge\bar{f}_{m}(h_{m};\bm{\theta}^{*},\mathbf{h}_{-m}).
\]
Similarly, $\text{\ensuremath{\forall}\ensuremath{\epsilon}\ensuremath{>0}}$,
when $\forall h_{m}>h_{m}^{*}$,
\[
\bar{f}_{m}(h_{m}+\ensuremath{\epsilon};\bm{\theta}^{*},\mathbf{h}_{-m})\ge\bar{f}_{m}(h_{m};\bm{\theta}^{*},\mathbf{h}_{-m}).
\]
To sum up, $\bar{f}_{m}(h_{m};\bm{\theta}^{*};\mathbf{h}_{-m})$ is
quasiconvex in $h_{m}$.

\section{Proof of Theorem 2}\label{app:proofthequa}

By Theorem \ref{thm:quatwosof}, we already have $\bar{f}_{m,1}(h_{m,1};\bm{\theta}^{*},\mathbf{H}_{m,1}^{-})$
is quasiconvex under the given condition, where $\mathbf{H}_{m,k}^{-}$
is a matrix representing $\mathbf{H}$ removing $h_{m,k}$. Then,
we generalize the proof to $K>1$, when $S_{k}(\mathbf{p}^{(i)};\mathbf{H})$
is chosen as the indicator function, \emph{i.e}., $\omega_{0}=1$.
As mentioned previously, $\forall m$, given $\mathbf{H}_{m,1}^{-}\succeq\mathbf{H}_{m,1}^{-}$
and obstacle indicator function (\ref{eq:esthidvircon}), values of
$\mathbb{I}_{k}^{(i)}(\mathbf{H})$ have two types: (i) $\mathbb{I}_{0}^{(i)}(\mathbf{H})=\mathbb{I}\{h_{m,1}<z_{m}^{(i)}\},\mathbb{I}_{1}^{(i)}(\mathbf{H})=\mathbb{I}\{h_{m,1}\ge z_{m}^{(i)}\}$;
(ii) independent of $h_{m,1}$. The case (i) is also known as measurement
samples which can only be blocked by $m$th grid and denote it as
\begin{equation}
\mathcal{L}_{m,K}(\mathbf{H})=\left\{ i\mid h_{j,K}<z_{j,K}^{(i)},\forall j\in\mathcal{B}^{(i)}\setminus\{m\}\right\} \label{eq:lmkforK}
\end{equation}
where $K=1$. In a general $K>1$ case, we denote a set $\mathcal{L}_{m,k}(\mathbf{H})$
for samples only blocked by $m$th grid and no more obstructed than
$k$th class obstacle
\[
\mathcal{L}_{m,k}(\mathbf{H})=\left\{ i\mid h_{j,l}<z_{j}^{(i)},\forall l\ge k,\forall j\in\mathcal{B}^{(i)}\setminus\{m\}\right\} .
\]
For quasiconvex of \ac{nlos} obstacle, \emph{i.e}., obstacle of the
most obscured propagation, we can similarly follow the pipeline of
previous $K=1$ situation, function (\ref{eq:formulnonoi}) becomes
\begin{equation}
\bar{f}_{m,K}(h_{m,K};\bm{\theta}^{*},\mathbf{H}_{m,k}^{-})=\frac{1}{N}\sum\limits _{i=1}^{N}[\bar{g}^{(i)}-\sum_{k=0}^{K}\gamma_{k}^{(i)}\mathbb{I}_{k}^{(i)}(\mathbf{H})]^{2}.\label{eq:formulsimK}
\end{equation}
Recall that $\bar{g}^{(i)}=\sum_{k=0}^{K}\gamma_{k}^{(i)}\mathbb{I}_{k}^{(i)}(\mathbf{H}^{*})$
is from (\ref{eq:channel-model-nonoi}), and we denote $\gamma_{k}^{(i)}=\beta_{k}^{*}+\alpha_{k}^{*}d(\mathbf{p}^{(i)})$
as the path loss for $d(\mathbf{p}^{(i)})$ and the $k$th degree
of signal obstruction, $\mathbb{I}_{k}^{(i)}(\mathbf{H})=\mathbb{I}\{\mathbf{p}^{(i)}\in\mathcal{D}_{k}(\mathbf{H})\}$.
Substitute $\bar{g}^{(i)}$ into (\ref{eq:formulsimK}). We obtain
\begin{align*}
\bar{f}(h_{m,K};\bm{\theta}^{*},\mathbf{H}_{m,k}^{-}) & =\frac{1}{N}\sum\limits _{i=1}^{N}[\sum_{k=0}^{K}\gamma_{k}^{(i)}\mathbb{I}_{k}^{(i)}(\mathbf{H}^{*})-\sum_{k=0}^{K}\gamma_{k}^{(i)}\mathbb{I}_{k}^{(i)}(\mathbf{H})]^{2}\\
 & =\frac{1}{N}\sum\limits _{i=1}^{N}[\sum_{k=0}^{K}\gamma_{k}^{(i)}(\mathbb{I}_{k}^{(i)}(\mathbf{H}^{*})-\mathbb{I}_{k}^{(i)}(\mathbf{H}))]^{2}.
\end{align*}
For term $\sum_{k=0}^{K}\gamma_{k}^{(i)}(\mathbb{I}_{k}^{(i)}(\mathbf{H}^{*})-\mathbb{I}_{k}^{(i)}(\mathbf{H}))$
in it: We have
\[
\sum_{k=0}^{K}\gamma_{k}^{(i)}(\mathbb{I}_{k}^{(i)}(\mathbf{H}^{*})-\mathbb{I}_{k}^{(i)}(\mathbf{H}))=\begin{cases}
A^{(i)}, & \textrm{if }\mathbb{I}_{k}^{(i)}(\mathbf{H}^{*})=\mathbb{I}_{k}^{(i)}(\mathbf{H}),\\
B^{(i)}, & \textrm{if }\mathbb{I}_{k}^{(i)}(\mathbf{H}^{*})\neq\mathbb{I}_{k}^{(i)}(\mathbf{H}).
\end{cases}
\]
where $A^{(i)}=0$ and $B^{(i)}\neq0$. Denote two sets $\mathcal{I}_{A}(\mathbf{H})=\{i\mid\mathbb{I}_{k}^{(i)}(\mathbf{H})=\mathbb{I}_{k}^{(i)}(\mathbf{H}^{*})\}$,
$\mathcal{I}_{B}(\mathbf{H})=\{i\mid\mathbb{I}_{k}^{(i)}(\mathbf{H})\neq\mathbb{I}_{k}^{(i)}(\mathbf{H}^{*})\}=\{1,2,\ldots,N\}\setminus\mathcal{I}_{A}(\mathbf{H})$
and similar to the proof of Theorem \ref{thm:quatwosof}, the smaller
the distance between $h_{m,K}$ and $h_{m,K}^{*}$ is, the larger
the set $\mathcal{I}_{A}(h_{m,K};\mathbf{H}_{m,k}^{-})$ will be,
\emph{i.e}., the more $\mathbb{I}_{k}^{(i)}(h_{m,K};\mathbf{H}_{m,k}^{-})=\mathbb{I}_{k}^{(i)}(h_{m,K}^{*};\mathbf{H}_{m,k}^{-})$.
Therefore, $\bar{f}_{m,K}(h_{m,K};\bm{\theta}^{*},\mathbf{H}_{m,k}^{-})$
is quasiconvex in $h_{m,K}$. We have the sets $\mathcal{L}_{m,K}(\mathbf{H})$
and class-$K$ obstacle heights $h_{m,K}$, where $m=1,2,\ldots,M$.
Followed by \ac{nlos} obstacle, the obstacle of less obscured propagation
is class-$K-1$ and the set is
\[
\mathcal{L}_{m,K-1}(\mathbf{H})=\left\{ i\mid h_{j,l}<z_{j}^{(i)},\forall l\ge K-1,\forall j\in\mathcal{B}^{(i)}\setminus\{m\}\right\} 
\]
in which $l=K-1,K$. On the condition of the sets $\mathcal{L}_{m,K}(\mathbf{H})$
and class-$K$ obstacle heights $h_{m,K}$, $\mathcal{L}_{m,K-1}(\mathbf{H})$
can be also written as
\[
\mathcal{L}_{m,K-1}(\mathbf{H})=\left\{ i\mid h_{j,K-1}<z_{j}^{(i)},\forall j\in\mathcal{B}^{(i)}\setminus\{m\}\right\} \cap\mathcal{L}_{m,K}(\mathbf{H})
\]
which is now also similar to the previous situation (\ref{eq:lmkforK}).
In other words, when we have the most obscured propagation situation,
the less obscured propagation situation can be recursively get. As
a result, for all $m,k$, quasiconvexity of $\bar{f}_{m,k}(h_{m,k};\bm{\theta}^{*},\mathbf{H}_{m,k}^{-})$
is proved and then $\bar{f}(\bm{\theta}^{*},\mathbf{H})$ is element-wise
quasiconvex for each individual element $h_{m,k}$.

\section{Proof of Theorem 3}\label{app:proofthedom}

First of all, the optimal solution $\mathcal{H}^{*}$ to minimize
function (\ref{eq:formulnonoi}) are intervals, because in the objective
function, variable $\mathbf{H}$ is only in an indicator function
taking $1$ by comparing the range of those values. Denote $\mathcal{M}_{m}$
as the index set of measurements $(\mathbf{p}^{(i)},y^{(i)})$ that
the link $(\mathbf{p}_{\text{u}}^{(i)},\mathbf{p}_{\text{d}}^{(i)})$
passes over the $m$th grid cell. Then, by the proof of Theorem \ref{thm:quaharbou},
similar to $\mathcal{L}_{m,k}(\mathbf{H})$, denote a set $\mathcal{Q}_{m,k}(\mathbf{H})$
for samples passing $m$th grid and no more obstructed than $k$th
class obstacle, \emph{i.e}.,
\begin{equation}
\mathcal{Q}_{m,k}(\mathbf{H})=\left\{ i\mid h_{j,l}<z_{j}^{(i)},\forall l\ge k,\forall j\in\mathcal{B}^{(i)},\forall i\in\mathcal{M}_{m}\right\} \label{eq:qmkdefforh}
\end{equation}
and the optimal set $\mathcal{Q}_{m,k}(\mathbf{H}^{*})$
\[
\mathcal{Q}_{m,k}(\mathbf{H}^{*})=\left\{ i\mid h_{j,l}^{*}<z_{j}^{(i)},\forall l\ge k,\forall j\in\mathcal{B}^{(i)},\forall i\in\mathcal{M}_{m}\right\} 
\]
So it follows that $\mathcal{Q}_{m,1}(\mathbf{H})\subseteq\mathcal{Q}_{m,2}(\mathbf{H})\subseteq\cdots\subseteq\mathcal{Q}_{m,K}(\mathbf{H}),\forall m,K>1.$
Recall the obstacle indicator function in (\ref{eq:esthidvircon}),
then it is obvious that $h_{m,1}\ge h_{m,2}\ge\cdots\ge h_{m,K},\forall m,K>1.$
Given the condition that $\mathbf{H}'\succeq\mathbf{H}^{*}$, we thus
have
\begin{equation}
\mathcal{Q}_{m,k}(\mathbf{H}')\subseteq\mathcal{Q}_{m,k}(\mathbf{H}^{*}),\forall m,k.\label{eq:lmkinlmks}
\end{equation}
From the proof of Theorem \ref{thm:quaharbou} and Fig. \ref{fig:indobj},
denote the $k$th class obstacle height estimation of $m$th grid
\begin{eqnarray}
\hat{h}_{m,k} & = & \sup\arg\min_{h_{m,k}}\bar{f}_{m,k}(h_{m,k};\bm{\theta}^{*},\mathbf{H}_{m,k}^{-})\nonumber \\
 & = & \min(z_{m}^{(i)}),i\in\mathcal{Q}_{m,k}(\mathbf{H}').\label{eq:hhatminz}
\end{eqnarray}
Therefore, by (\ref{eq:lmkinlmks})\textendash (\ref{eq:hhatminz}),
\begin{eqnarray*}
\hat{h}_{m,k} & = & \min(z_{m}^{(i)}),i\in\mathcal{Q}_{m,k}(\mathbf{H}')\\
 & \ge & \min(z_{m}^{(i)}),i\in\mathcal{Q}_{m,k}(\mathbf{H}^{*})=h_{m,k}^{*}.
\end{eqnarray*}
Similarly, for $k=K,K-1,\ldots,1$ and for all $m$, we have the same
results, so we can have $\hat{\mathbf{H}}\succeq\mathbf{H}^{*}$.

As for the $\hat{h}_{m,k}(\bm{\theta}^{*},\mathbf{H}'')\geq\hat{h}_{m,k}(\bm{\theta}^{*},\mathbf{H}')$
under $\mathbf{H}''\succeq\mathbf{H}'\succeq\mathbf{H}^{*}$: When
$\mathbf{H}'$ taking value of $\mathbf{1}H_{\text{max}}$, $\mathbf{H}''$
can only equal to $\mathbf{1}H_{\text{max}}$ and
\begin{eqnarray*}
\hat{h}_{m,k}(\bm{\theta}^{*},\mathbf{H}') & = & \min(z_{m}^{(i)}),i\in\mathcal{Q}_{m,k}(\mathbf{1}H_{\text{max}})\\
 & = & \hat{h}_{m,k}(\bm{\theta}^{*},\mathbf{H}'')
\end{eqnarray*}
so $\hat{h}_{m,k}(\bm{\theta}^{*},\mathbf{H}'')\geq\hat{h}_{m,k}(\bm{\theta}^{*},\mathbf{H}')$
and $H_{\text{max}}\geq\hat{h}_{m,k}(\bm{\theta}^{*},\mathbf{H}')$.
For arbitrary $\mathbf{H}'$, when $\mathbf{H}''\succeq\mathbf{H}'\succeq\mathbf{H}^{*}$
is satisfied, by (\ref{eq:qmkdefforh}) $\mathcal{Q}_{m,k}(\mathbf{H}'')\subseteq\mathcal{Q}_{m,k}(\mathbf{H}')$
and similar to (\ref{eq:hhatminz})
\begin{eqnarray*}
\hat{h}_{m,k}(\bm{\theta}^{*},\mathbf{H}'') & = & \min(z_{m}^{(i)}),i\in\mathcal{Q}_{m,k}(\mathbf{H}'')\\
 & \ge & \min(z_{m}^{(i)}),i\in\mathcal{Q}_{m,k}(\mathbf{H}')=\hat{h}_{m,k}(\bm{\theta}^{*},\mathbf{H}')
\end{eqnarray*}
thus it always holds.

\end{appendices}

\bibliographystyle{IEEEtran}
\bibliography{IEEEabrv,StringDefinitions,JCgroup,ChenBibCV}

\end{document}